\documentclass[11pt]{article}

\usepackage{graphicx}
\usepackage{float}
\usepackage{amssymb}
\usepackage{comment}

\usepackage[T1]{fontenc}
\usepackage[latin9]{inputenc}
\usepackage{array}
\usepackage{multirow}
\usepackage{amsmath}
\usepackage{setspace}
\usepackage{esint}
\usepackage[authoryear]{natbib}
\providecommand{\tabularnewline}{\\}
\usepackage{amsthm}
\usepackage{thmtools}\usepackage{latexsym}
\usepackage{xcolor}
\usepackage{tikz,pgf}

\newtheorem{theorem}{Theorem}
\newtheorem{lemma}[theorem]{Lemma}

\newcommand{\openr}{\hbox{${\rm I\kern-.2em R}$}}
\newcommand{\openn}{\hbox{${\rm I\kern-.2em N}$}}

\newcommand{\norm}[1]{\left\lVert#1\right\rVert}

\usepackage{JASA_manu}

\begin{document}

\title{Causal Inference for Social Network Data}
\author{Elizabeth L. Ogburn\thanks{Department of Biostatistics, Johns Hopkins Bloomberg School of Public Health, Baltimore, MD, USA}, Oleg Sofrygin\thanks{Kaiser Permanente Division of Research, 2000 Broadway, Oakland, CA, 94612, USA}, Iv\'{a}n D\'{i}az\thanks{Division of Biostatistics and Epidemiology, Weill Cornell Medicine, New York, NY, USA}, and Mark J. van der Laan\thanks{Department of Biostatistics, University of California Berkeley, 2121 Berkeley Way, Berkeley, CA, 94720, USA} }

\maketitle


\pagenumbering{gobble} 
\begin{center}
\textbf{Abstract}
\end{center}
We describe semiparametric estimation and inference for causal effects using observational data from a single social network. Our asymptotic results are the first to allow for dependence of each observation on
a growing number of other units as sample size increases. In addition, while previous methods have 
implicitly permitted only one of two possible sources of dependence among
social network observations, we allow for both dependence due to
transmission of information across network ties and for dependence
due to latent similarities among nodes sharing ties. We propose new causal effects that are specifically of interest
in social network settings, such as interventions on network ties and network structure. We use our methods to reanalyze an influential and controversial study that estimated causal peer effects of obesity using social network data from the Framingham Heart Study; after accounting for network structure we find no evidence for causal peer effects.  

\vspace*{.3in}


\newpage
 \pagenumbering{arabic}

\section{Introduction \label{sec:Introduction}}

There is increasing interest in identifying
and estimating causal effects in the contexts of social networks,
that is, causal effects that one individual's behavior, treatment assignment,
beliefs, or health outcome could have on their social contacts'
behaviors, exposures, beliefs, or health statuses. But methodology has not kept apace with interest in causal inference using
data from individuals connected in a social network, and many researchers
have resorted to using inappropriate statistical methods to analyze this
new type of data. There have been
a number of high profile articles that use standard methods like generalized
linear models (GLM) and generalized estimating equations (GEE) to
attempt to infer causal peer effects from network data (e.g. \citealp{christakis2007spread,christakis2008collective,christakis2010social}),
and this work has inspired several research programs that study peer
effects using the same statistical methods \citep{ali2010social,cacioppo2009alone,madan2010social,rosenquist2010spread,wasserman2013comment}.
These methods have come under considerable criticism from
the statistical community \citep{cohen2008obesity,lyons2011spread,shalizi2011homophily}, in part because
these statistical models are not equipped to deal with dependence
across individuals \citep{ogburn2014vaccines}.

Recently, many researchers have developed methods specifically designed to deal with interference--the effect that one subject's treatment may have on others' outcomes--and other forms of causal dependence among subjects. 
However, the methods developed in this context generally
require observing multiple independent groups of units, i.e.\ multiple independent networks, or else they require treatment
to be randomized. Ideally, we would like to be able to perform inference
even when all observations are sampled from a single social network
and when exposures cannot be randomized. Aside from \cite{vanderlaan2012}, upon which our work builds, \citet{tchetgen2020auto} is the only other proposed solution to this problem of which we are
aware. The approach of \citet{tchetgen2020auto} is quite different from ours, primarily because
it assumes that the outcomes of interest comprise a single realization
of a specific type of Markov random field over the network.  This corresponds to certain types of equilibrium distributions and is incompatible with the
traditional causal data generating mechanisms that we work with in
this paper, namely causal structural equation models and directed
acyclic graph (DAG) models (for a discussion of these compatibility
issues see \citealp{lauritzen2002chain, ogburn2020causal}).

We build upon recent work by \cite{vanderlaan2012} that proposed estimators for traditional causal estimands using data from a single collection of interconnected units in which each
unit is independent of all but a small number of other
units, with asymptotic results relying on the number of dependent units being
fixed as the total number of units goes to infinity and with dependence due solely to direct transmission.   
We prove a new central limit theorem for network dependent data, an implication of which is that the estimators proposed by \cite{vanderlaan2012} are consistent and asymptotically normal under more realistic forms of dependence and more general asymptotic regimes.  
As far as we are aware our methods are the first to accommodate latent, unstructured dependence; previous methods (including \citealp{vanderlaan2012} and 
 \citealp{tchetgen2020auto}) are appropriate when all dependence is due to observable direct transmission of information, treatments, or outcomes along network ties. This is a crucial contribution as latent variable dependence is likely to be present in all but the most pathological social network settings. We also introduce novel causal estimands that are specifically of interest in network settings and propose the first conditions that permit causal inference for interventions on network structure. Together these contributions comprise a broad framework for causal inference using observational data from a single social network.
 In this paper we focus on the theory and only briefly touch on estimation; details regarding the implementation and computation of our estimation procedure can be found in a series of companion papers focused on implementation and computation \citep{sofryginTechreport,sofrygin2017conducting,sofrygin2018single}. An R package is also available \citep{sofrygin2015tmlenet}. 

\section{Background and setting}\label{sec:setting}

\subsection{Motivating example}

 In order to demonstrate the importance of principled methods designed to handle the complexity of observational social network data, in Section \ref{sec:data} we reanalyze the Framingham Heart Study data used in \cite{christakis2007spread}, which purported to find evidence that obesity is socially contagious.  The FHS is an epidemiologic cohort study that was originally designed to study cardiovascular epidemiology and is comprised of over $15,000$ participants from the town of Framingham, Massachusetts and neighboring communities. In addition to its important role in epidemiology, the FHS plays a uniquely influential role in the study of social networks and peer effects. In the early 2000s, researchers Christakis and Fowler (CF) discovered an untapped resource buried in the FHS data collection tracking sheets: information on social ties that, combined with existing data on connections among the FHS participants, allowed them to reconstruct the (partial) social network underlying the cohort.  They then
leveraged this social network data to study peer effects for obesity~\citep{christakis2007spread},
smoking~\citep{christakis2008collective}, and happiness~\citep{fowler2008dynamic}.  Researchers have since used the same methods as Christakis and Fowler to study peer effects in the FHS and in many other social network settings (e.g.  \citealp{trogdon2008peer,fowler2008dynamic,rosenquist2010spread}).   

We will refer to a toy example based on FHS throughout. For the purposes of our example we will assume that Framingham is a closed community and that we have data on all members of the community and knowledge of their connections with one another.
Suppose, as has been suggested by some researchers
\citep{fowler2008dynamic}, that happiness exhibits peer effects, and let $Y$ be a continuous measure of happiness, measured simultaneously on each of the $n$ individuals in the FHS network. Below we will describe causal estimands of interest in the context of this motivating example. 

\subsection{Networks and causal estimands} \label{SEM}\label{sub:estimands}

A network is a collection of units, or nodes, and information about the presence
or absence of pairwise ties between them. The presence of a tie between
two units indicates that the units share some kind of a relationship; for example in the FHS ties include familial relatedness, friendship, and shared place of work.
Some types of relationships, like familial relatedness, are mutual; others,
like friendship, may go in only one direction. For
simplicity we will assume all networks are undirected in what follows,
but our methods are equally applicable to directed networks. In an
undirected network, the \textit{degree} of a node is the number of
ties it has. The \textit{alters} of node $i$ are the nodes with which $i$ shares ties. Let $A_{ij}\equiv I\left\{ \mbox{subjects }i\mbox{ and }j\mbox{ share a tie}\right\} $ and $A_{ii}=1$.  
The matrix $\mathbf{A}$ with entries $A_{ij}$ is the \emph{adjacency matrix}
for the network. We will assume throughout that each node is associated with a vector of random
variables, $O_{i}$, including an outcome $Y_{i}$, covariates $C_{i}$,
and an exposure or treatment variable $X_{i}$, all possibly indexed
by time $t$. Throughout we will use bold letters to represent $n$-vectors of random variables, e.g. $\mathbf{Y}$ denotes $(Y_1,...,Y_n)$. 

In applications across the social, political,
and health sciences, researchers are interested in ascertaining the
presence of and estimating causal interactions across alter-ego pairs.
Is there \textit{interference}, i.e.\ does the treatment of subject $i$ have
a causal effect on the outcome of subject $j$ when $i$ and $j$
share a network tie? Are there \textit{peer effects}, i.e.\ does the outcome
of subject $i$ have a causal effect on its alters' future outcomes? Implicit in many of these causal
inquiries is the notion that causal effects only transmit along network ties. 
It is  crucial that the network be completely and accurately measured, since missing ties can result in unmeasured confounding or missing exposure variables. For example, $C_i$, e.g. $i$'s baseline attitude about mental health, may be a confounder not just for the effect of $X_i$ on $Y_i$ but also for the effect of $X_j$ on $Y_j$ if it affects $j$'s treatment uptake and $j$'s happiness.
If the existence of a tie between units $i$ and $j$ is missing in our data, we may fail to control for confounders $C_i$ of causal effects on $Y_j$. We could also fail to account for elements $X_i$ of the exposure affecting $Y_j$, resulting in exposure misclassification. 

Let $Y_{i}(\mathbf{x}^{*})$
denote the \emph{potential} or \emph{counterfactual} outcome of individual $i$ in a hypothetical
world in which the exposure vector $\mathbf{X}$ was set to $\mathbf{x}^*$, e.g. the happiness outcome we would have observed for individual $i$ if all FHS subjects had received meditation training. 
We index potential outcomes with the vector of exposures because, in contrast with i.i.d.\ settings, 
  $Y_i$ may be causally affected by the exposures of other nodes. Additionally, note that these potential outcomes may not be identically distributed.
We will focus throughout on network-wide estimands of the form  $E\left[\bar{Y}_{n}(\mathbf{x}^*)\right]$, where $\bar{Y}_{n}(\mathbf{x}^*)=\frac{1}{n} \sum_{i=1}^n Y_{i}(\mathbf{x}^{*})$: the expected potential outcome, averaged across network nodes, in the hypothetical
world where $\mathbf{X}$ is set to $\mathbf{x}^*$.
The distribution of $Y_{i}(\mathbf{x}^{*})$ depends on which nodes $j$ are causally related to $i$ -- that is, it depends on the observed adjacency matrix $\mathbf{A}$ and, through that, on $n$. 
Therefore, all of our estimands and estimators condition on the adjacency matrix $\mathbf{A}$ and on $n$. 
This same conditioning is implicit in other work on causal inference for data from a single network, since no methods currently exist to jointly model the network structure given by $\mathbf{A}$ \emph{and} the random variables associated with the nodes, but this is rarely made explicit. Estimands of the form  $E\left[\bar{Y}_{n}(\mathbf{x}^*)\right]$ are \emph{data-dependent} or \emph{data adaptive} \citep{hubbard2016statistical}, meaning that their value depends on the observed data sample through $\mathbf{A}$ and $n$. 
This changes but does not undermine their interpretation. Consider a hypothetical intervention to make FHS participants happier. Using data from FHS we may not be able to learn about the expected average potential outcome under this intervention for the U.S. population (unless the effect of the intervention does not depend on $n$ or $\mathbf{A}$), but we may be able to learn about the expected average potential outcome under this intervention for other mid-sized, middle class, New England suburbs. 

We will also, at times, discuss causal estimands conditional on covariates $\mathbf{C}$, e.g. $E\left[\bar{Y}_{n}(\mathbf{x}^*) |\mathbf{C}\right]$.  As $n \rightarrow \infty$, $E\left[\bar{Y}_{n}(\mathbf{x}^*) |\mathbf{C}\right] \rightarrow E\left[\bar{Y}_{n}(\mathbf{x}^*)\right]$, however for finite $n$ conditional estimators have smaller variance and inference about the conditional estimand
cannot be interpreted as inference about the marginal estimand. Inference about conditional estimands could be informative about other, similar networks that share a common distribution of $\mathbf{C}$. For example, if $C$ is age, then we might be able to use inference about  $E\left[\bar{Y}_{n}(\mathbf{x}^*) |\mathbf{C}\right]$ to learn about the expected average potential outcome of the intervention in other mid-sized, middle class New England suburbs with a similar age distribution to Framingham. If $C$ is not an effect modifier, then causal effects are the same conditional on or marginalized over $\mathbf{C}$ and conditioning on $\mathbf{C}$ does not limit our ability to generalize causal effects. 


We consider three families of causal interventions in this paper: static, dynamic, and stochastic. Static interventions set $X_i$ to a user-specified, fixed value $x^*$ that does not depend on unit $i$'s characteristics. The literature on causal inference with interference has largely focused on static interventions to date. This family includes, for example, the average potential outcome in a world in which everyone received exposure $X=1$ compared to a world in which everyone received exposure $X=0$.  In the interference literature, this effect is known as an \emph{overall} effect \citep{hudgens2008toward,tchetgen2010causal}. It also includes versions of \emph{direct} and \emph{indirect} effects, which quantify the effect of a unit's own treatment holding other treatments constant and the effect of others' treatments holding a unit's own treatment constant, respectively, though the most common definition of direct and indirect effects requires the presence of multiple independent groups of individuals and will not apply in our single network setting.
A dynamic intervention assigns exposures as a user-specified, deterministic function
of covariates $\mathbf{C}$, e.g. an FHS happiness intervention that assigns subjects to talk therapy if their baseline happiness is below a certain threshhold.  Exposure is deterministically specified conditional on covariates
but is but allowed to depend ``dynamically'' on covariates.  
A stochastic intervention \citep{munoz2012population,haneuse2013estimation,young2014identification} assigns exposures as a user-specified, random function. This allows exposure to be a stochastic function of covariates, e.g. an FHS happiness intervention that assigns subjects to talk therapy with a probability that depends on their baseline happiness. In Section \ref{sec:Extensions} we use stochastic interventions to define some kinds of interventions on network structure. 


\subsection{Inferential and asymptotic framework\label{sub:Asymptotics}}


Throughout we assume that we have observed a complete network of size $n$. Our methods are not appropriate for a random sample from a superpopulation because such a random sample would miss crucial confounders and exposures, as described in Section \ref{SEM} above. The data may exhibit \emph{latent variable dependence} -- unstructured dependence resulting from latent
traits that are more similar for observations that are close in the network
than for distant observations -- up to a distance of two network ties. Latent variable dependence will be present in data sampled from a
network whenever observations from nodes that are close to one another
are more likely to share unmeasured traits than are observations from
distant nodes. Homophily, or the tendency of people who share similar
traits to form network ties, is a paradigmatic example of latent variable
dependence.
Additionally, in networks, ties present opportunities to transmit causal effects or information, and
\emph{dependence due to direct transmission} will be present whenever a subject's treatments, outcomes, or covariates may affect
their alters' treatments, outcomes, or covariates. 

We will show in Section \ref{sub:identification} that, despite this dependence, the causal estimands of interest are identified by functionals of the observed data distribution $P(\mathbf{C},\mathbf{X},\mathbf{Y})$, which depends on $n$ and $\mathbf{A}$. This identifying functional is our target of inference. It is a parameter of the observed data distribution for a network of size $n$, i.e.\ a parameter  \emph{of the data generating distribution that gave rise to the data at hand}. It is an unknown parameter rather than an observed quantity because the data we see comprise a single, random draw from $P(\mathbf{C},\mathbf{X},\mathbf{Y})$. 
In i.i.d.\ settings inference licenses extrapolation to populations for which the data are representative, in the sense of being drawn from (nearly) the same underlying data generating distribution. The same principle applies here: we might want to estimate a parameter of the FHS data generating distribution because we think that it is informative about the social networks found in other mid-sized, middle class New England suburbs. In i.i.d.\ settings, if a causal effect does not depend on a particular aspect of the data-generating distribution, then we are licensed to extrapolate to other populations that differ from the data with respect to that aspect. Similarly, here, if we have reason to believe that the causal effect of interest is homogeneous in $n$ and relevant features of $\mathbf{A}$, then we might be justified in using an analysis of the FHS data to extrapolate to other small, mid-sized, and large middle-class New England suburbs. All of our estimands and estimators condition on $n$ and $\mathbf{A}$, but this only matters for interpretation if the estimands do in fact vary with $n$ and $\mathbf{A}$.

In theory, we could learn about the parameter of interest from an infinite number of draws (each of size $n$) from the underlying distribution. This is the inferential framework that corresponds to most previous work on causal inference with interference and/or social network data, in which it is typically assumed that multiple i.i.d.\ groups or networks are observed and asymptotics are in the number of  groups. In this paper we tackle the more challenging and more realistic setting of data from a single interconnected network, and we perform inference using data from a single draw of size $n$.  We will use $n \rightarrow \infty$ asymptotics in the service of finite sample inference. That is, we will prove that our estimators are asymptotically normal under general and relatively realistic conditions; when these conditions are met and for large enough $n$, this licenses the use of asymptotic approximations to perform inference the parameters indexing the data generating distribution for finite $n$.  
This, too, is analogous to the i.i.d.\ setting, where asymptotic results are used to license finite sample approximations for large enough $n$.

Research indicates that social networks often have the small-world
property (sometimes referred to as the ``six degrees of separation''
property), meaning that the average distance between two nodes is
small \citep{watts1998collective}. Therefore distances in real-world
networks may grow slowly (or not at all!) with sample size, making it difficult to apply existing dependent data theory and methods to this setting and necessitating a new central limit theorem that we prove in the Appendix.  Formalizing asymptotic growth of network-generating models as $n$ goes to infinity is an active area of research, especially for networks that, like social networks, have sparse limits \citep{caron2017sparse} and is beyond the scope
of this paper. Because we condition on the graph itself through the adjacency matrix $\mathbf{A}$, we do not need to include the graph-generating mechanism in our asymptotic regime, obviating many of the complex issues surrounding asymptotic growth of networks.
We take for granted a sequence of networks with increasing
$n$, represented by adjacency matrices $\mathbf{A}_n$, such that key features of the network topology, e.g. degree distribution or clustering, are preserved. The two crucial features of this sequence are (1) that same structural equation model holds for each network in the sequence 
and (2) that the maximum degree grows slower than rate $\sqrt{n}$. 
This asymptotic regime is consistent with realistic social networks.  In particular, social networks tend to have heavy-tailed degree distributions,
with most nodes having low degree but a non-trivial proportion of
nodes having high degree, with the maximum degree dependent on the
size of the network, resulting in asymptotically sparse networks \citep{newman2003social}. 
This is incompatible with the assumption of bounded degrees,
which has been used in previous methods for inference about observations
sampled from a single social network. See Section \ref{sec:hubs} for a discussion of inference in the presence of highly connected "hub" nodes. 


\subsection{Review of influence functions for non-i.i.d.\ data and data-dependent estimands}

In order to employ semiparametric theory in this  dependent data setting, we draw on classical semiparametric theory as presented in \cite{bickel1998efficient}, 
and on work that extending these results to non-i.i.d.\ settings \citep{mcneney2000application, vanderlaan2012, janssen2013convolution, van1996convolution, hubbard2016statistical}. Many aspects of the classical theory, though developed and presented for i.i.d.\ data, are agnostic to the i.i.d.\ assumption. Key differences are that we index our parameters with $n$; we allow rates of convergence to be any function $C_n$ of $n$, possibly slower than $\sqrt{n}$; and we do not assume an i.i.d.\ CLT. 

Let $P$ be the true distribution of the observed data $\mathbf{O}=(O_1,...,O_n)$ and let $\mathcal{M} \ni P$ be a model that may place some restrictions on $P(\mathbf{O})$; $\mathcal{M}$ could be nonparametric, semiparametric, or parametric. Let $\{P^h_{\epsilon}:\epsilon\in(-\delta,\delta)\} \subset \mathcal{M}$ be a rich class of one-dimensional parametric submodels, indexed by a direction $h$ across a set $H$, such that $P^h_{\epsilon=0}=P$ for all $h$. Let $S^h=\frac{d}{d\epsilon}\log P^h_{\epsilon}(O)|_{\epsilon=0}$ be the corresponding score function of the $h$-specific path. 
Then the pathwise derivative of parameter $\psi_{n}$ at $P$ along the path defined by $P^h_{\epsilon}$ is given by $\frac{d}{d\epsilon}\psi_{n}(P^h_{\epsilon})$ at $\epsilon =0$, and by the Riesz representation theorem $
 \frac{d}{d\epsilon}\psi_{n}(P^h_{\epsilon})|_{\epsilon=0}=E[\varphi_{n}(O)S^h(O)]$
 for some mean zero function $\varphi{}_{n}(O)$.  We say that $\psi_{n}$ is pathwise differentiable at $P$ if this equation holds for all submodels in the class defined above, and then $\varphi{}_{n}(O)$ is a gradient for $\psi_{n}$. 
Pathwise differentiability requires only that the pathwise derivative be a bounded linear operator from the tangent space (i.e. the closure of the linear span of the scores $S^h$) to the real line, and therefore does not rely on data being i.i.d..  
In general each gradient for $\psi_{n}$ is the influence function of an estimator of $\psi_{n}$, and in subsequent sections we will use these terms interchangeably. 
In a nonparametric model ${\cal M}$ that places no restrictions on $P(\mathbf{O})$, there can be no more than one gradient for any parameter $\psi_n$, but in semiparametric or parametric models $\psi_n$ will have many gradients. If $\varphi_n$ is a gradient for $\psi_n$ in a model $\mathcal{M}$ it will also be a gradient for $\psi_n$ in any submodel of $\mathcal{M}$, and in particular the nonparametric gradient will be a gradient in any semiparametric or parametric model for $P(\mathbf{O})$. 

We say that an estimator $\hat\psi(\mathbf{O})$ of $\psi_n$ has influence function $\varphi_n$ if $\hat\psi(\mathbf{O})-\psi_n=\frac{1}{n}\sum_{i=1}^{n}\varphi_n(\{O_{i}\})+o_{P}(\sqrt{C_{n}})$, where $\{O_{i}\}$ represents a subset of $\mathbf{O}$ but may include variables indexed by $j$ for $j\neq i$ in addition to those that are indexed by $i$. The influence function of an estimator describes its first order behavior;  the limiting distribution of an asymptotically linear estimator is given by the limiting distribution of $\frac{1}{n}\sum_{i=1}^{n}\varphi_n(\{O_{i}\})$. 
Under a central limit theorem, $\sqrt{C_{n}}\frac{1}{n}\sum_{i=1}^{n}\varphi_n(\{O_{i}\})\overset{D}{\rightarrow}N(0,\sigma^{2})$ where the asymptotic variance $\sigma^{2}$ of an estimator is given by the asymptotic variance of $\frac{1}{n}\sum_{i=1}^{n}\varphi_n(\{O_{i}\})$. Note that although the parameter of interest can depend on $n$, its limiting distribution does not. Because an influence function depends on $\psi_n$ and has mean $0$ at the truth, it can often be used as an estimating function; an estimator that solves $\frac{1}{n}\sum_{i=1}^{n}\varphi_n(\{O_{i}\})=0$ will have influence function $\varphi_n$ (provided that estimates of any unknown nuisance parameters satisfy appropriate regularity conditions). 

\section{Methods}
\label{sec:methods}

The estimating procedure that we describe in this section is based on 
 \citet{vanderlaan2012}, but we generalize the results
to a broader class of causal effects, to more general and pervasive
forms of dependence among observations, and to more realistic asymptotic regimes. 
We focus throughout on single time-point treatments. Longitudinal interventions
are also possible under the theory introduced here, in particular because the longitudinal extension of our structural equation model implies that, for each time point and after conditioning on the past, the data exhibit no more dependence than the single-time point case. We leave the details for future work. We state our results under the assumption that all variables take
values in discrete sets. Analogous results are valid for other types
of random variables: it is straightforward to extend our notation
and central limit theorem to continuous covariates and outcomes,  
but continuous treatments, which are not standard in causal inference, are more challenging (see \citealp{vanderlaan2012}). For estimands that condition on $\mathbf{C}$, our theoretical results hold immediately if $\mathbf{C}$ is discrete but may require additional assumptions for continuous $\mathbf{C}$.

\subsection{Structural equation model} \label{sub:SEM}

We use a causal structural equation model (SEM) to define the causal
effects of interest but note that
analogous definitions may be achieved within the potential outcome
framework \citep{pearl2012causal}.  

Let $K_{i}=\sum_{j=1}^{n}A_{ij}$ be the degree of node $i$. The degree of $i$ and the degrees of
$i$'s alters may be included in the covariate vector $C_{i}$. Define $\mathbf{Y}=(Y_{1},...,Y_{n})$ and $\mathbf{C}$ and $\mathbf{X}$
analogously. We assume that the data are generated by sequentially evaluating the
following set of equations: 
\begin{align}
C_{i} & =f_{C}\left[\varepsilon_{C_{i}}\right] & i=1,\ldots,n\nonumber \\
X_{i} & =f_{X}\left[\left\{ C_{j}:A_{ij}=1\right\} ,\varepsilon_{X_{i}}\right] & i=1,\ldots,n\nonumber \\
Y_{i} & =f_{Y}\left[\left\{ X_{j}:A_{ij}=1\right\} ,\left\{ C_{j}:A_{ij}=1\right\} ,\varepsilon_{Y_{i}}\right] & i=1,\ldots,n,\label{eq:original SEM}
\end{align}
where $f_{C}$, $f_{X}$, and $f_{Y}$ are unknown and unspecified
functions that may depend on $K_i$ and $\varepsilon_{i}=(\varepsilon_{C_{i}},\varepsilon_{X_{i}},\varepsilon_{Y_{i}})$
is a vector of exogenous, unobserved errors for individual $i$. The errors may be correlated across units, as described below. This set of equations represents an observational setting; all our results apply equally to experimental settings in which $f_{X}$ is constant or does not depend on $\mathbf{C}$.  
Time ordering
is a fundamental component of a structural causal model. We assume that $C$ is first drawn for all units, so that, in addition
to $C_{i}$, the other components of the vector $\mathbf{C}$--corresponding to $i$'s social contacts--may affect the
value of $X_{i}$, and similarly for $X$.

In addition, we make the following assumptions on the error terms from the SEM:
\begin{align}
 &\mbox{The vectors } (\varepsilon_{X_{1}},...,\varepsilon_{X_{n}}),(\varepsilon_{Y_{1}},...,\varepsilon_{Y_{n}}), \mbox{ and } (\varepsilon_{C_{1}},...,\varepsilon_{C_{n}}) \mbox{ are independent},\tag{A1}\label{eq:independent errors}\\
 & \varepsilon_{X_{1}},...,\varepsilon_{X_{n}} \mbox{ are identically distributed and } \varepsilon_{Y_{1}},...,\varepsilon_{Y_{n}} \mbox{ are identically distributed,}\tag{A2a}\label{eq:(X,Y) i.i.d.}\\
 & \varepsilon_{X_{i}}\perp\varepsilon_{X_{j}}\mbox{ and }\varepsilon_{Y_{i}}\perp\varepsilon_{Y_{j}} \mbox{ for }i,j\mbox{ s.t.} {A_{ij}=0} \mbox{ and } \exists!k \mbox{ with }A_{ik}=A_{kj}=1\tag{A2b}\label{eq:(X,Y)ind}\\
 & \varepsilon_{C_{1}},..., \varepsilon_{C_{n}}\mbox{ are identically distributed, and}\tag{A3a}\label{eq:C identically distributed}\\
 & \varepsilon_{C_{i}}\perp\varepsilon_{C_{j}}\ensuremath{\mbox{ for }i,j\mbox{ s.t. \ensuremath{A_{ij}=0} and \ensuremath{\ensuremath{\exists!k}} with }A_{ik}=A_{kj}=1.}\tag{A3b}\label{eq:independent C errors}
\end{align}
A DAG corresponding to the SEM in (1) is given in Figure 1(a).  Each node represents an $n$-vector; independence relationships across individuals (i.e. assumptions (A2) and (A3)) are not encoded in this DAG.

This SEM or DAG model encodes the assumption that $\mathbf{C}$ suffices to
control for confounding of the effect of $\mathbf{X}$ on $\mathbf{Y}$;
this is a version of the \emph{conditional ignorability} or \emph{no unmeasured confounding} assumption that is typically made in i.i.d.\ settings. In particular, any latent variable dependence affects $\mathbf{X}$
and $\mathbf{Y}$ separately; in general a latent variable that affects
$\mathbf{X}$ and $\mathbf{Y}$ jointly is an unmeasured confounder and constitutes a violation
of this assumption.  Assumptions (\ref{eq:(X,Y)ind}) and (\ref{eq:independent C errors}) relax the typical assumption of independent errors to permit latent variable dependence up to a distance of
two network ties. 
Assumption (A3) can be omitted if attention is restricted to causal effects conditional on $\mathbf{C}$. 

Although our main result, given in Theorem 1 below, holds for inference
under assumptions (\ref{eq:(X,Y) i.i.d.})--(\ref{eq:independent C errors}),
some estimation strategies are available only when stronger versions
of assumptions (\ref{eq:(X,Y)ind}) and (\ref{eq:independent C errors})
hold. We therefore introduce alternative assumptions 
\begin{align}
 & \varepsilon_{X_{1}},...,\varepsilon_{X_{n}}\mid\mathbf{C}\mbox{ are i.i.d.\ and \ensuremath{\varepsilon_{Y_{1}},...,\varepsilon_{Y_{n}}\mid\mathbf{C},\mathbf{X}} are i.i.d.},\mbox{ and}\tag{A4}\label{eq:i.i.d. x and y}\\
 & \varepsilon_{C_{i}},i=1,...,n\mbox{, are i.i.d.}\tag{A5}\label{eq:i.i.d. c}
\end{align}
These assumptions, which were made in \cite{vanderlaan2012}, are consistent with dependence due to direct transmission but not latent variable dependence.

\begin{figure}
\begin{centering}
\begin{tikzpicture}[>=stealth, node distance=1.3cm]
    \tikzstyle{format} = [ very thick, circle, minimum size=5.0mm,
	inner sep=0pt]
    \tikzstyle{square} = [very thick, rectangle, draw]

		\begin{scope}[yshift=-.50cm]
		   \path[very thick, ->]
                        node[format] (x) {$\mathbf{X}$}
                        node[format, above right of=x, xshift=.25cm, yshift=.5] (c) {$\mathbf{C}$}
                        node[format, below right of=c, xshift=.5cm] (y) {$\mathbf{Y}$}

                        (x) edge[black] (y)
                        (c) edge[black] (x)
                        (c) edge[black] (y)
                        
                        
                    			
			node[below of=c, xshift=0cm, yshift=-1.8cm] (l) {(a)}

			;
	\end{scope}
        \begin{scope}[xshift=8.0cm]
                \path[very thick, ->]
                        node[format] (w) {$\mathbf{W}$}
                        node[format, above of=w, yshift=.5] (c) {$\mathbf{C}$}
                        node[format, below left of=w, xshift=-1cm] (x) {$\mathbf{X}$}
                        node[format, below of=w] (v) {$\mathbf{V}$}
                        node[format, below right of=w, xshift=1cm] (y) {$\mathbf{Y}$}

                        (x) edge[red] (v)
                        (w) edge[black] (x)
                        (w) edge[black] (y)
                        (c) edge[red] (w)
                        (v) edge[black] (y)

                        

		node[below of=v, xshift=0cm, yshift=0cm] (l) {(b)}
                    
                ;
        \end{scope}
\end{tikzpicture} 
\par\end{centering}
\caption{(a) DAG corresponding to the SEM in Equation (1); (b) DAG corresponding to the SEM in Equation (A6). The red arrows are deterministic. 
 }
\label{fig:graph} 
\end{figure}

In principle, the
models defined by assumptions (\ref{eq:independent errors})-(\ref{eq:independent C errors})
or by assumptions (\ref{eq:independent errors}), (\ref{eq:i.i.d. x and y}),
and (\ref{eq:i.i.d. c}) suffice to nonparametrically identify many causal estimands of interest. However, in practice (and in order to facilitate the definition and identification of some kinds of dynamic and stochastic estimands) we may need to make simplifying
assumptions on the forms of $f_{X}$ and $f_{Y}$. This is done by
considering summary functions $s_{C}$ and $s_{X}$ and random variables
$W_{i}=s_{C,i}\left(\left\{ C_{j}:A_{ij}=1\right\} \right)$ and $V_{i}=s_{X,i}\left(\left\{ X_{j}:A_{ij}=1\right\} \right)$
such that the model may be written as 
\begin{align*}
C_{i} & =f_{C}\left[\varepsilon_{C_{i}}\right] & i=1,\ldots,n\\
X_{i} & =f_{X}\left[W_{i},\varepsilon_{X_{i}}\right] & i=1,\ldots,n\\
Y_{i} & =f_{Y}\left[V_{i},W_i,\varepsilon_{Y_{i}}\right] & i=1,\ldots,n. \tag{A6} \label{summary functions}
\end{align*}
For example, $s_{C,i}\left(\left\{ C_{j}:A_{ij}=1\right\} \right)=\left(C_{i},\sum_{j:A_{ij}=1}C_{j}\right)$
implies that the exposure and outcome of node $i$ only depend on $i$'s
own covariate value and on the sum of the covariate values of $i$'s alters. Analogously, $s_{X,i}\left(\left\{ X_{j}:A_{ij}=1\right\} \right)=\left(X_{i},\sum_{j:A_{ij}=1}X_{j}\right)$
is an example of a summary function for $\mathbf{X}$. It may be the case that $f_X$ and $f_Y$ depend on different summary functions of $\mathbf{C}$, $s_{C,X}$ and $s_{C,Y}$ respectively. In this case, without loss of generality, we define $s_C=W$ to be equal to $(s_{C,X},s_{C,Y})$. The natural choice of summary functions may not have the same length for all $i$. For example, it could be natural to define the summary function $s_C$ to be $C_{i}$ for units with no alters and $(C_{i},\sum_{j:A_{ij}=1}C_{j})$ for units with alters. In order to enforce that  $s_{C,i}$ and $s_{X,i}$ have the same length for all $i$, we set the length of each $s_{C,i}$ and $s_{X,i}$ to its maximum over $i$, and fill in any empty entries with the value $undefined$. 
For convenience
we use the notation $s_{C,i}(\mathbf{C})$ and $s_{X,i}(\mathbf{X})$
below; however, this notation should not undermine the important fact
that $W_{i}$ can only depend on the subset $\left\{ C_{j}:A_{ij}=1\right\} $ of $\mathbf{C}$
and $V_{i}$ can only depend on the subset
 $\left\{ X_{j}:A_{ij}=1\right\} $ of  $\mathbf{X}$,
as these are the only components of $\mathbf{C}$ and $\mathbf{X}$
that are parents of $\mathbf{X}$ and $\mathbf{Y}$, respectively,
in the network-as-structural-causal-model.
For notational convenience, in what follows we augment the observed
data random vector with $V_{i}$ and $W_{i}$, recognizing that these
are deterministic functionals of $C_{i}$ and $X_{i}$, defined by
$s_{X,i}$ and $s_{C,i}$, and are therefore technically redundant. 
A DAG corresponding to the SEM in (A6) is given in Figure 1(b). 
Although this SEM implies that we can use $\mathbf{W}$ and $\mathbf{V}$ in place of $\mathbf{C}$ and $\mathbf{X}$ for identification and estimation of causal effects, for generality and clarity we will continue to use $\mathbf{C}$ and $\mathbf{X}$ wherever we do not require the assumption of fixed dimensionality. 

Note that, although the variance-covariance structure of the SEM given
in (\ref{eq:original SEM}) is affected by the dependence allowed
in (A2b) and (A3b), the mean structure is unaltered by the choice
of assumptions (A2) and (A3) or (A4) and (A5), ruling out the possibility that any latent sources of dependence introduce confounding. In particular, while the SEM allows limited forms of homophily to induce dependence it rules out confounding due to homophily, where latent similarities affect both the exposure and the outcome and thereby induce confounding. Ruling this out is a strong and often unrealistic assumption \citep{shalizi2011homophily}. Because the mean structure is unaffected by the dependence permitted by assumptions (A2) and (A3), any estimator that is unbiased under (A4) and (A5) will remain unbiased when these are relaxed to (A2) and (A3). In Section \ref{sub:identification}
we discuss nonparametric identification of causal parameters, which relies on the assumption of no unmeasured confounding but
is agnostic to the choice of the weaker or stronger independence assumptions.

\subsection{Definition and nonparametric identification of causal effects} \label{sub:identification}

We first define notation that we will use throughout the remainder
of the paper for functionals of the distribution of $\mathbf{O}$.
Let  $p_{Y}(\mathbf{y}|\mathbf{v,w})=P\left(\mathbf{Y=y}|\mathbf{V=v,W=w}\right)$ and
 $p_{Y,i}(y|v,w)$ $=$ 
$P\left(Y_{i}=y|V_{i}=v,W_i=w\right)$.  Let $p_{C}(\mathbf{c})=P\left(\mathbf{C}=\mathbf{c}\right)$.  We will use $g$ to denote the propensity score distributions: $g(\mathbf{x}|\mathbf{w})=P\left(\mathbf{X}=\mathbf{x}\mid\mathbf{W}=\mathbf{w}\right)$ and
$g_{i}(x|w)=P\left(X_{i}=x|W_{i}=w\right)$, and $h$ for the corresponding distributions of $\mathbf{W}$ and $\mathbf{V}$: 
$h_{i}(v|w)=P\left(V_{i}=v|W_i=w\right)$ and $h_{i}(v,w)=P\left(V_{i}=v,W_i=w\right)$. Let $h^*$ denote this distribution under the intervention of interest: $h_{i}^*(v|w)=P\left(V_{i}^{*}=v|W_i=w\right)$ and $h_{i}^*(v,w)=P\left(V_{i}^{*}=v,W_i=w\right)$. This will be degenerate for static and deterministic interventions but not for stochastic interventions. 
Note that $h_{i}^*$ are determined by $g$, $p_{C}$, and the user-specified intervention and is therefore an observed data quantity. Finally,
$m(v,w)=\sum_{y}y\,p_{Y}(y|v,w)$ is the conditional expectation of $Y$
given $V=v,W=w$.

A hypothetical intervention on $\mathbf{X}$ replaces  $g(x|c)$ with a new, user-specified function $g^{*}$; under Assumption (\ref{summary functions}) this is equivalent to replacing $h$ with $h^{*}$. Equivalently, the intervention replaces $f_{X}$ in the SEM with a new, user-specified function.  For example, a deterministic intervention that sets $X_{i}$ to
a user-given value $x_{i}^{*}$ for $i=1,...,n$ is given by 
\begin{align*}
C_{i} & =f_{C}\left[\varepsilon_{C_{i}}\right] & i=1,\ldots,n\\
X_{i} & =x_{i}^{*} & i=1,\ldots,n\\
Y_{i}^* & =f_{Y}\left[v_{i}^*,W_i,\varepsilon_{Y_{i}}\right] & i=1,\ldots,n,
\end{align*}
where $\mathbf{x}^{*}=(x_{1}^{*},\ldots,x_{n}^{*})$. Here $Y_{i}^{*}$
denotes the potential or counterfactual outcome of individual $i$ in a hypothetical
world in which $P(\mathbf{X}=\mathbf{x}^{*})=1$. Analogously, $v_{i}^{*}=s_{X,i}(\mathbf{x}^{*})$
is a counterfactual  variable in a hypothetical world in which
$P(\mathbf{X}=\mathbf{x}^{*})=1$. Note that, although $v_{i}^*$
is counterfactual, its value is determined by the user-specified value $\mathbf{x}^{*}$,
and it is therefore known. In the case of a stochastic intervention, $X_{i}^{*}$ and $V_{i}^{*}$ are random rather than fixed variables but their distributions are still known by design. In particular the distribution of $V_{i}^{*}|W_{i}$ is given by $h^{*}$. Although any intervention on $\mathbf{X}$ induces an intervention on $\mathbf{V}$, not all conceivable interventions on $\mathbf{V}$ will be compatible with the observed network. For example, suppose $V_{i}$ is the average of $X_{i}$ and $X_{j}$ for $A_{ij}=1$, with $X$ binary. Even if $0$ and $1$ are both in the support of $V$, it would not be possible to simultaneously assign $V_{i}^{*}=0$ and $V_{j}^{*}=1$ for $i,j$ s.t. $A_{ij}=1$. In order to avoid such contradictions we recommend defining interventions in terms of $\mathbf{X}$ and $g$ prior to expressing the induced intervention on $\mathbf{V}$ and $h$. 
The causal parameter of interest 
is the expected average potential outcome $E\left[\bar{Y}^{*}_{n}\right]$ defined in Section \ref{sub:estimands}. 

In addition to the assumption of no unmeasured confounding, identification of $E\left[\bar{Y}^{*}_{n}\right]$
relies on the positivity assumption that, for all $i$,
\begin{align}
P(V_i=v|s_{C,i}(\mathbf{C})=s_{C,i}(\mathbf{c}))>0 & \mbox{\,\,\ for all $\mathbf{c}$ in the support of $\mathbf{C}$ and for all \ensuremath{v} in the range of \ensuremath{V^*}.}\tag{A7}\label{eq:positivity}
\end{align}
This assumption states that, within levels of $\mathbf{C}$, the values
of $V$ determined by the hypothetical intervention
have positive probability under the observed data generating distribution.
Now the causal parameter $E\left[\bar{Y}^{*}_{n}\right]$ for a static intervention
is nonparametrically identified as follows:
\begin{align}
E\left[\bar{Y}^{*}_{n}\right]	&=\frac{1}{n}\sum_{i=1}^{n}E[Y_{i}^*] \notag \\
	&=\frac{1}{n}\sum_{i=1}^{n}\sum_{\mathbf{c}}E[Y_{i}^*|\mathbf{C=c}]p_{C}(\mathbf{c}) \notag \\
	&=\frac{1}{n}\sum_{i=1}^{n}\sum_{\mathbf{c}}E[Y_{i}|V_i=s_{X,i}(\mathbf{x}^{*}),W_i=s_{C,i}(\mathbf{c})]p_{C}(\mathbf{c}) \notag \\
	&=\frac{1}{n}\sum_{i=1}^{n}\sum_{\mathbf{c}}m\{s_{X,i}(\mathbf{x}^{*}),s_{C,i}(\mathbf{c})\}p_{C}(\mathbf{c}) \label{eq:identification form 2}
\end{align}
where the third equality follows from the assumption of conditional unconfoundedness. Assumption (A7) ensures that the conditioning event has positive probability and therefore that the conditional expectation is well-defined. Assumption (A6) is not required for nonparametric identification but we have used it above to simplify notation (and we will rely on it for estimation and inference).  
This identification result is equivalent to  
\begin{equation}
E\left[\bar{Y}^{*}_{n}\right]=\frac{1}{n}\sum_{i=1}^{n}E\left[m(v_{i}^{*}, W_i)\right]=\frac{1}{n}\sum_{i=1}^{n}\sum_{w}m(v_i^*,w)h_i^{*}(v_i^*,w).\label{eq:identification form 1}
\end{equation}
From \eqref{eq:identification form 2}, it is clear that the conditional
 parameter $E\left[\bar{Y}^{*}_{n}\mid\mathbf{C}=\mbox{\textbf{c}}\right]$
is identified by $\frac{1}{n}\sum_{i=1}^{n}m\{s_{X,i}(\mathbf{x}^{*}),s_{C,i}(\mathbf{c})\}$.

Throughout, we will denote the functional of the observed data that identifies a causal estimand of interest as $\psi_n$. This is the statistical parameter about which we would like to perform inference.

\subsection{Estimation}\label{sub:Estimation}

Estimation and inference for $E\left[\bar{Y}^{*}_{n}\right]$ require
a statistical model $\mathcal{M}$ for the distribution of the observed
data $P(\mathbf{O})$. That is, $\mathcal{M}$ is a collection of
distributions over $\mathbf{O}$ of which one element is the true
data-generating distribution. The only restriction placed on the observed data distribution by Assumptions (A1)-(A3b) is that, for for $i,j:A_{ij}=0\text{ and }\text{\ensuremath{\exists}}!k\text{ with }A_{ik}=A_{kj}=1$ we have that $C_{i}\perp C_{j}$, $X_{i}\perp X_{j}|\mathbf{C}$, and $Y_{i}\perp Y_{j}|\mathbf{C},\mathbf{X}$; Assumptions (A4) and (A5) -- and previous work such as \cite{vanderlaan2012} and \cite{sofryginTechreport}-- restrict these conditional independences to hold for all $i,j$.  Assumption (A6) implies that $\mathbf{Y}\perp\mathbf{C,X}|\mathbf{W,V}$ and $\mathbf{X}\perp\mathbf{C}|\mathbf{W}$.
Under assumption (A6) the probability
distribution of the observed data may be factorized as 
\begin{equation}
P\left(\mathbf{O}=\mathbf{o}\right)=p_{C}\left(\mathbf{c}\right)g(\mathbf{x}|\mathbf{w})p_{Y}(\mathbf{y}|\mathbf{v,w}),\label{eq:factorization}
\end{equation}
suggesting that $\mathcal{M}$ requires three components: a model
for $p_{C}$, a model for $g$, and a model for $p_{Y}$. Furthermore, the identification results in (\ref{eq:identification form 2}) and (\ref{eq:identification form 1}) indicate that the identifying functional $\psi_n$ depends
on $p_Y$ only through $m$, and under assumption (A6) it depends on $g$ only through $h$. The empirical distribution $\hat{p}_{C}$
can be used throughout to nonparametrically average with respect to $p_{C}$, but,
when $\mathbf{C}$ is high-dimensional, $h$ and $m$ may not be nonparametrically
estimable at rates of convergence that are fast enough to satisfy
the regularity conditions of Theorem 1 (see Appendix). Therefore,
we will specify a statistical
model $\mathcal{M}=\mathcal{M}_{h}\times\mathcal{M}_{m}$, where $\mathcal{M}_{h}$
is a collection of conditional distributions for $\mathbf{V}$ given $\mathbf{W}$ such
that the true conditional distribution is a member, and $\mathcal{M}_{m}$
is a collection of conditional expectations
of $\mathbf{Y}$ given $\mathbf{V}$ and $\mathbf{W}$ such that the true conditional expectation of $\mathbf{Y}$ is a member. Because $P\left(\mathbf{C}=\mathbf{c}\right)$, $g(\mathbf{x}|\mathbf{w})$, and $p_{Y}(\mathbf{y}|\mathbf{v,w})$ are all variation independent, $\mathcal{M}_{h}$ does not restrict $P\left(\mathbf{C}=\mathbf{c}\right)$ or $p_{Y}(\mathbf{y}|\mathbf{v,w})$ and $\mathcal{M}_{m}$ does not restrict $P\left(\mathbf{C}=\mathbf{c}\right)$ or $g(\mathbf{x}|\mathbf{w})$. In principle we do not need to be able to estimate $h$ and $m$ at parametric rates in order for our estimator to achieve asymptotic normality at parametric rates (see Appendix for specific rate conditions). Nonparametric estimation procedures for dependent data are an active area of research \citep{bibaut2021sequential} but may still be difficult to implement under latent variable dependence, and in our data analysis and simulations we rely on parametric models for $h$ and $m$.

Under assumptions (A4)-(A7) an influence function for $\psi_n$, evaluated at a fixed value $\mathbf{o}$ of $\mathbf{O}$, was derived by \citet{vanderlaan2012} and its sample average is given by
\begin{equation}
D_n(\mathbf{o})=  \dfrac{1}{n}\sum_{i=1}^{n}\left(E\left[m\left(V_{i}^{*},W_i\right)\mid\mathbf{C}=\mathbf{c}\right]-\psi_n+\frac{\bar{h}^*(v_{i},w_i)}{\bar{h}(v_{i},w_i)}\left\{ y_{i}-m\left(v_{i},w_i\right)\right\} \right)\label{eq:EIF-1}
\end{equation}
where $\bar{h}(v_{i},w_i)=\frac{1}{n}\sum_{j=1}^{n}h_{j}(v_{i},w_i)$, $\bar{h}^*(v_{i},w_i)=\frac{1}{n}\sum_{j=1}^{n}h_j^*(v_{i},w_i)$, and
$v_{i}=s_{X,i}(\mathbf{x})$. For a deterministic intervention, $V_{i}^{*}=s_{X,i}(\mathbf{x}^{*})$ is not random; we discuss the case of random $V_i^*$ in Section \ref{sec:Extensions}.
$D_n(\mathbf{o})$ has expected value equal to $0$ at the true
$\psi_n$; this fact can be used to generate unbiased estimating equations
for $\psi_n$. \citet{vanderlaan2012} showed that estimators with this influence
function are doubly robust: the right hand side of Equation (\ref{eq:EIF-1})
has expected value equal to $0$ if $m(\cdot)$ is replaced with an
arbitrary functional of $V$ or if $h(\cdot)$ is replaced with an
arbitrary functional of $W$, as long as one of the two remains correctly
specified. 
This implies that
an estimating equation based on Equation (\ref{eq:EIF-1}) will be unbiased
for $\psi_n$ if either model $\mathcal{M}_{m}$ for $m(\cdot)$ or
model $\mathcal{M}_{h}$ for $h(\cdot)$ is correctly specified, i.e.\
contains the truth, even if one is not. 
See Section 6 of \cite{vanderlaan2012} for the derivation of the influence function and the proof of double robustness. 

Although \cite{vanderlaan2012} derived this influence function under a model defined by assumption (A4) (in addition to (A6) and (A7), which we assume throughout), the same derivation holds under  
assumptions (A2) and (A3). This is because, as we argued above, the same functional $\psi_n$ identifies $E[\bar{Y}_{i}^{*}]$ under either set of assumptions, and furthermore any estimator that is (asymptotically) unbiased under (A4) and (A5) will remain so when these are relaxed to (A2) and (A3), implying that any influence function under the model defined by (A4) and (A5) is also an influence function under (A2) and (A3). In Theorem 1 we will prove that the resulting estimator is CAN under  assumptions (A2) and (A3).

Below we propose a targeted maximum loss-based estimator (TMLE) of $\psi_n$. All of the results that follow are equally applicable to a standard estimating equation approach in which estimating equations for the parameters indexing a model for $m$ and a model for $h$ are stacked with the influence function estimating equation for $\psi_n$ (see, e.g., \citealp{kennedy2016semiparametric}). More details about implementation can be found in companion papers focused on implementation and computation \citep{sofryginTechreport,sofrygin2017conducting,sofrygin2018single} and an R package is available \citep{sofrygin2015tmlenet}.

TMLE is a general template for estimation
of smooth parameters in semi- and nonparametric models. The estimation
algorithm is constructed to solve an influence function
estimating equation (hence the asymptotic equivalence with the estimating equation approach). 
In our setting, a TMLE is constructed using three elements: (i) a valid loss
function $L$ for the outcome regression model $m$, (ii) initial
working estimators $\hat{m}$ of $m$ and and $\hat{h}$ of $h$,
and (iii) a parametric submodel $m_{\epsilon}$ of $\mathcal{M}$,
the score of which corresponds to a particular component of the score
based on the influence function $D_n(\mathbf{o})$ and such
that $m_{\epsilon =0}=m(\cdot)$. The TMLE is then defined by an iterative procedure
that, at each step, estimates $\epsilon$ by minimizing the empirical
risk of the loss function $L$ at $m_{\epsilon}$. An updated estimate
is then computed as $\hat{m}_{\hat{\epsilon}}$, and the process is
repeated until convergence. The TMLE is the estimator obtained in
the final step of the iteration. The result of the previous iterative
procedure is that, at the final step, the influence function
estimating equation is solved. For more details about targeted maximum
likelihood estimation, see \citet{van2011targeted}. In the present setting,
the TMLE for $\psi_n$ based on $D_n(\mathbf{o})$ requires only one
iteration for convergence \citep{van2011targeted}. 
Initial parametric estimators $\hat{m}$ and $\hat{h}$ of $m$ and $h$ may be found
through maximum likelihood or loss-based estimation methods like standard regression models. (The proof of Theorem 1 also suffices to prove that an m-estimator for either of the nuisance models will be CAN for its expectation.)  Alternatively, under a conditional independence structure analogous to that implied by assumptions (A1), (A4), and (A5), \citet{benkeser2018online} showed that super learning \citep{van2007super} can be used to nonparametrically estimate the nuisance models. The empirical distribution $\hat{p}_{C}$ is
used to marginalize with respect to $p_{C}$. 

We propose using a direct estimate $\hat{\bar{h}}$ of $\bar{h}$ that optimizes the log likelihood function $\sum_{i=1}^{n}\log\bar{h}(V_{i}|W_{i})$
as if the pooled sample $(V_{i},W_{i})$ were i.i.d.\ It can be shown
that this results in a valid loss function for $\bar{h}$, even for
dependent observations $(V_{i},W_{i})$ \citep{vanderlaan2012,sofryginTechreport}. Similarly,
one can construct a direct estimator $\hat{\bar{h}}_{x^{*}}$ of $\bar{h}_{x^{*}}$ by first
creating a sample $(V_{i}^{*},W_{i})$ and then directly optimizing
the log likelihood function $\sum_{i=1}^{n}\log\bar{h}_{x^{*}}(V_{i}^{*}|W_{i})$,
as if the pooled sample $(V_{i}^{*},W_{i})$ were i.i.d. We perform estimation of the conditional mixture density $\bar{h}$ using a conditional histogram approach, previously described for i.i.d.\ data in \citet{munoz2011super}. The approach relies on fitting the conditional hazards of individual bins from the support of $V_{i}$ (given $W_{i}$) using separate parametric logistic regression models. 
In our highly-dependent network settings, the operational characteristics of the direct estimator of $\bar{h}$ are unclear. 
However, we believe that the enormous computational advantages offered by this direct estimation route, along with the encouraging results obtained from our extensive simulations, merit the description of this estimator. We also realize that more theoretical work is needed to justify and improve upon this direct approach. For additional details and simulation results that demonstrate the performance of the direct estimation approach for mixture density $\bar{h}$, we refer to \citet{sofryginTechreport,sofrygin2017conducting,sofrygin2018single}.

Now the TMLE of $\psi_n$ is computed
as follows: 
\begin{enumerate}
\item Define the auxiliary weights $H_{i}$ as the ratio of estimated densities
of $V^{*},W$ and $V,W$ evaluated at the observed value $W_{i}$. Compute
the auxiliary weights as 
\[
H_{i}=\frac{\hat{\bar{h}}^*(V^*_{i},W_i)}{\hat{\bar{h}}(V_{i},W_i)}.
\]

\item Compute initial predicted outcome values $\hat{Y}_{i}\equiv\hat{m}(V_{i},W_i)$
and predicted potential outcome values $\hat{Y}_{i}^{*}\equiv\hat{m}(V_{i}^{*},W_i)$
evaluated at the counterfactual value $V_{i}^{*}$. Under a static intervention $V_i^*$ is the degenerate random variable $s_{X,i}(\mathbf{x}^{*})$. 
\item Construct a TMLE model update $\hat{m}_{\hat{\epsilon}}$ of $\hat{m}$
by running a weighted intercept-only logistic regression model with
weights $H_{i}$ defined in step (1), $Y_{i}$ as the outcome and
including $\hat{Y}_{i}$ as an offset. That is, define $\hat{\epsilon}$
as the estimate of the intercept parameter $\epsilon$ from the following
\emph{weighted} logistic regression model 
\[
\mbox{logit}\hat{m}_{\epsilon}(v,w)=\mbox{logit}\hat{m}(v,w)+\epsilon,
\]
where $\mbox{logit}(x)=\log\left(\frac{x}{1-x}\right)$. 
\item Compute updated predicted potential outcomes $\tilde{Y_{i}}^{*}$
as the fitted values of the regression from step (c), evaluated at $v^{*}$
rather than $v$ (that is, at $\hat{Y}_{i}^{*}$ instead of $\hat{Y}_{i}$):
\[
\tilde{Y_{i}}^{*}=\mbox{expit}\{\mbox{logit}\hat{Y}_{i}^{*}+\hat{\epsilon}\},
\]
where $\mbox{expit}(x)=\frac{1}{1+e^{-x}}$, i.e.\, the inverse of
the $\mbox{logit}$ function. 
\item Compute the TMLE $\hat{\psi_n}$ as 
\[
\hat{\psi_n}=\frac{1}{n}\sum_{i=1}^{n}\tilde{Y_{i}}^{*}.
\]

\end{enumerate}
The TMLE is doubly robust: it will be consistent for $\psi_n$
if either the working model $\hat{h}$ for $h$ (in this case comprised of models for $\bar{h}$ and $\bar{h}_{x^*}$) or the working model
$\hat{m}$ for $m$ is correctly specified. This resulting estimator
remains CAN for $\psi_n$ under assumptions (A2) and (A3) or
(A4) and (A5), and the same procedure can be used to estimate the
parameter conditional on $\mathbf{C}$.

\subsection{Asymptotic normality} \label{sub:Asymptotic-normality}

We consider an asymptotic regime in which $K_{i}$ may
grow as $n\rightarrow\infty$ and prove that $\hat{\psi}_n$ converges to a normal limiting distribution under assumptions (A2) and (A3), which implies the same result under Assumptions (A4) and (A5). The proof relies on an analysis of the influence function of $\hat{\psi}_n$ and is therefore agnostic to estimating procedure (i.e.\ it holds for the TMLE and estimating equation approaches). However, it requires that one of the models for $m$ and $h$ be correctly specified and converge to the truth at rate $\sqrt{C_n}$ or that both models converge to the the truth such that the product of their rates of convergence is $\sqrt{C_n}$, e.g. both converge to the truth at rate $C_n^{-1/4}$ (see the regularity conditions in the Appendix).

\textbf{Theorem 1:\label{Theorem 1} }\textit{Suppose that $K_{max,n}^{2}/n\rightarrow0$
as $n\rightarrow\infty$, where $K_{max,n}=max_{i}\{K_{i}\}$
for network size $n$. Under assumptions (A2), (A3), (A6), (A7), and regularity conditions
(see Appendix), 
\begin{align*}
 & \sqrt{C_{n}}\left(\hat{\psi_n}-\psi_n\right)\overset{d}{\longrightarrow}N(0,\sigma^{2}),
\end{align*}
for some finite $\sigma^{2}$ and for some $C_n$ such that $n/K_{max,n}^{2}\leq C_{n}\leq n$.}

The asymptotic variance $\sigma^{2}$ of $\hat{\psi_n}$ is given by the asymptotic variance of $D_n(\mathbf{O})$, the sample average of the influence function of the estimator. 
The proof of Theorem 1 is in the Appendix. Broadly, the proof has two parts:
first, to show that the second order terms in the expansion of $\hat{\psi_n}-\psi_n$
are stochastically less than $1/\sqrt{C_{n}}$, and second, to show
that the first order terms converge to a normal distribution when
scaled by a factor of order $\sqrt{C_{n}}$. The proof that the second
order terms are stochastically less than $1/\sqrt{C_{n}}$ is an extension
of the empirical process theory of \citet{van1996weak} and follows
from the proof in \citet{vanderlaan2012}. For
the proof that the first order terms converge to a normal distribution,
we rely on Stein's method of central limit theorem proof \citep{stein1972bound}.
Stein's method allows us to derive a bound on the distance between
our first order term (properly scaled) and a standard normal distribution;
this bound depends on the degree distribution $K_{1},...,K_{n}$.
We show that this bound converges to $0$ as $n\rightarrow\infty$
under regularity conditions and our running assumption that $K_{max,n}^{2}=o(n)$.

When all nodes have the same number of ties, i.e.\ $K_{i}=K_{max,n}$
for all $i$, then the rate of convergence will be given by $\sqrt{C_{n}}=\sqrt{n/K_{max,n}^{2}}$.
When $K_{max,n}$ is bounded above as \textit{$n\rightarrow\infty$}, as in \citet{vanderlaan2012},
 the rate of convergence will be $\sqrt{n}$. When $K_{max,n}\rightarrow\infty$
but some nodes have fewer than $K_{max,n}$ ties, the exact rate of
convergence is between $\sqrt{n/K_{max,n}^{2}}$ and $\sqrt{n}$ 
but is difficult or impossible to determine analytically, as it may
depend intricately on the structure of the network. The inferential
procedures that we describe below do not require knowledge of the
rate of convergence.

In Section \ref{sec:hubs},
below, we discuss settings in which the conditions for this theorem
fail to hold, and ways to recover valid inference for conditional
estimands in some of these settings. 

\subsection{Inference}
\label{sec:inference}

An asymptotically valid 95\% confidence interval for $\psi_n$ is given by $\hat{\psi_n}\pm1.96\sigma/\sqrt{C_{n}}$.
In practice neither $\sigma$ nor $C_{n}$ are likely to be known,
but available variance estimation methods estimate the variance of
$\hat{\psi_n}$ directly, incorporating the rate of convergence
without requiring it to be known a priori. 

In principle, an estimate of the variance of the $\hat{\psi_n}$ can always be obtained by the plug-in estimator of the variance of the influence function, which depends on the observed data only through $m(\cdot)$, $g(\cdot)$, and $p_C(\mathbf{c})$ (see Appendix). When dependence is due to direct transmission, that is under assumptions (A1), (A4), and (A5), $C$ is i.i.d.\ and  $p_C(\mathbf{c})$ can be estimated with the empirical distribution of $C$. We prove in the Appendix that the plug-in estimator using the empirical average of the square of the
influence function, substituting $\hat{\psi_n}$ for $\psi_n$ and the
fitted values from the working models $\hat{h}$ and $\hat{m}$ for
$h$ and $m$, is consistent under correct specification of models for both $m$ and $h$. 
Although the estimator $\hat{\psi_n}$ is doubly robust this variance estimator is not and may be anticonservative if one,
but not both, of the models for $m$ and $h$ is correctly
specified. Using flexible or non-parametric specifications for these
models increases opportunities to estimate both consistently; this is feasible using i.i.d.\ methods when dependence is due to direct transmission because the dependent variables in each of the two nuisance models are conditionally i.i.d. 
For a detailed discussion of how to implement this variance estimator, see \citet{sofryginTechreport}.

An alternative approach to estimate the variance of $\hat{\psi}_n$
under assumptions (A1), (A4), and (A5) is to employ the following
version of a parametric bootstrap, which might offer improvements
in finite-sample performance over the previously described approach.  
For each
of $B$ bootstrap iterations, indexed by $b=1,\ldots,B$, first $n$
covariates $\mathbf{C}^{b}=(C_{1}^{b},\ldots,C_{n}^{b})$ are sampled
with replacement, then a 
model fit $\hat{g}$ is applied to sampling of $n$ exposures $\mathbf{X}^{b}=(X_{1}^{b},\ldots,X_{n}^{b})$,
followed by a sample of $n$ outcomes $\mathbf{Y}^{b}=(Y_{1}^{b},\ldots,Y_{n}^{b})$
based on the existing outcome model fit $\hat{m}$. The corresponding bootstrap summaries $W_{i}^{b}$
and $V_{i}^{b}$, for $i=1,\ldots,n$, are constructed by applying
the summary functions $s_{C}$ and $s_{X}$ to $\mathbf{C}^{b}$ and
$\mathbf{X}^{b}$, respectively. This bootstrap sample
is then used to obtain the predicted values from the existing auxiliary
covariate fit $(\hat{\bar{h}}_{x^{*}}/\hat{\bar{h}})(V_{i}^{b},W_i^b)$,
for $i=1,\ldots,n$, followed by a bootstrap-based fitting of $\epsilon$,
and finally, evaluation of bootstrap TMLE. Note that the TMLE model
update is the only model fitting step needed at each iteration of
the bootstrap, which significantly lowers the computational burden
of this procedure. The variance estimate is then obtained by taking
the empirical variance of bootstrap TMLE samples $\hat{\psi_n}^{b}$. Because the parametric bootstrap relies on known or assumed independences, and because only the TMLE model (i.e.\ not the full likelihood) is fit at each iteration, this procedure consistently estimates the variance of the first order terms in the expansion of $\hat{\psi_n}-\psi_n$, and we prove in the Appendix that the higher order terms are asymptotically neglible.
However, due to dependence across observations, one
must be judicious with applications of the bootstrap. For example,
the parametric bootstrap procedure described above requires conditional
independence of $X_{i}$ given $W_{i}$ and $Y_{i}$ given $(V_{i},W_i)$,
along with the consistent modeling of the corresponding factors of
the likelihood. It may seem natural to sample $V_{i}$ directly from
its corresponding auxiliary model fit, but this is likely to result
in an anti-conservative variance estimates, since the conditional
independence structure assumed for $\mathbf{V}$ is unlikely to hold by virtue of its construction
as a summary measure of the network. 

When latent variable dependence is present, that is under assumptions
(A1) through (A3), the empirical distribution may not consistently estimate functionals of $p_C(\mathbf{c})$ at rate $\sqrt{C_n}$. This is a problem for both methods of variance estimation; we describe three possible strategies for overcoming this challenge but acknowledge future research is needed to devise better solutions.  The first option is to use a version of block bootstrap to estimate the required functionals of $p_C(\mathbf{c})$. Bootstrap methods for network data are an area of active research and beyond the scope of this paper 
but we note that, under our dependence assumptions, $\{\mathbf{C}_I\} \perp \{\mathbf{C}_J\}$ for sets of indices $I$ and $J$ such that no node in $I$ is connected to any node in $J$ by a path of fewer than 3 ties. This is an m-dependence structure and well-established methods for block-bootstrap for m-dependent data are immediately applicable \citep{lahiri2003resampling}. However, unlike standard m-dependence settings where the underlying topology is Euclidean, it could be computationally challenging to identify independent blocks in network data. We leave implementation of this procedure for future research. A second option is to postulate a parametric model for the required functionals of $p_C(\mathbf{c})$. 
Correctly specifying the joint distribution of $\mathbf{C}$ may be challenging in many settings. Finally, the third option is to restrict attention to conditional estimands, for which variance estimation does not require estimating functionals of $p_C(\mathbf{c})$.  A simple plug-in estimator is available for the variance of the conditional influence function $D_{n}^C(\mathbf{o})$ (see the Appendix and \citealp{vanderlaan2012}) and this is the approach that we take in our simulations and data analysis. However, our theoretic results may require additional assumptions when $\mathbf{C}$ is continuous.

\section{Extensions} \label{sec:Extensions}

In this section we extend the estimation procedure to two causal effects
of great interest in the context of social networks: social contagion,
or peer effects, and interventions
on the network structure itself, i.e.\ interventions on $\mathbf{A}=\left[A_{ij}:i,j\in\{1,\ldots,n\}\right]$
where, as above, $A_{ij}\equiv I\left\{ \mbox{subjects }i\mbox{ and }j\mbox{ share a tie}\right\} $.

\subsection{Dynamic and stochastic interventions}

A dynamic intervention assigns exposures as a user-specified, deterministic
function of covariates. We operationalize this as the replacement
of $h(v|w)$ with a new, user-specified, function $h^{*}$ that depends
on $w$ but is nonrandom. A stochastic intervention assigns exposures
as a user-specified, random function. We operationalize this as the
replacement of $h(v|w)$ with a new, user-specified, random distribution
$h^{*}$ that may depend on $w$. In contrast, a static intervention
replaces $h$ with a constant function.

It is sometimes more natural to think about dynamic and stochastic
interventions in terms of an intervention SEM that modifies $f_{X}$,
replacing it with a user-specified function of $W$. As long as the
intervention SEM adheres to Assumption (\ref{summary
functions}), an intervention on $f_{X}$ induces an intervention
distribution $h^{*}$ -- this induced distribution is required for
our estimation results to hold. Alternatively, we can imagine intervening on the functional form
of $s_{X}$. In particular, we can define two different summary functions:
$s_{C,X}^{*}(\cdot)$, a user-specified functional describing the
hypothetical dependence of $X$ on $\mathbf{C}$, and $s_{X}^{*}(\cdot,\cdot)$,
a user-specified functional describing the hypothetical dependence
of $Y$ on $\mathbf{X}$. They are denoted by an asterisk because
they index hypothetical interventions rather than realized data-generating
mechanisms. Let $W_{X,i}^{*}=s_{C,X,i}^{*}(\mathbf{C})$ and $V_{i}^{*}=s_{X,i}^{*}(\mathbf{X}^{*})$.
Then the intervention SEM for this stochastic intervention is given
by 
\begin{align}
C_{i} & =f_{C}\left[\varepsilon_{C_{i}}\right] & i=1,\ldots,n\nonumber \\
X_{i}^{*} & =f_{X}\left[W_{X,i}^{*},\varepsilon_{X_{i}}\right] & i=1,\ldots,n\nonumber \\
Y_{i}^{*} & =f_{Y}\left[W_{i},V_{i}^{*},\varepsilon_{Y_{i}}\right] & i=1,\ldots,n.\label{eq: intervene on sX, sY-1}
\end{align}
This can be interpreted as an intervention where, for each $x^{*}$
in the support of $X$ and for $i=1,...,n$, $X_{i}$ is set to $x^{*}$
with probability $P\left[X=x^{*}|W=s_{C,X,i}^{*}(\mathbf{C})\right]$
and $V_{i}$ is set to $s_{X,i}^{*}(\mathbf{x}^{*})$ deterministically
for each possible realization \textbf{$\mathbf{x}^{*}$}. Because
$Y$ depends on $\mathbf{X}$ only through $V$, this is equivalent
to an intervention that sets $V_{i}$ to $v$ with probability $P\left[\mathbf{X}\in\left\{ \mathbf{x}^{*}:s_{X,i}^{*}(\mathbf{x}^{*})=v\right\} \mid\mathbf{W}=\mathbf{s}_{C,X}^{*}(\mathbf{C})\right]$,
where $\mathbf{s}_{C,X}^{*}(\mathbf{C})=\left(s_{C,X,1}^{*}(\mathbf{C}),...,s_{C,X,n}^{*}(\mathbf{C})\right)$.
That is, it is equivalent to an intervention on $h$ with the intervention
distribution given by $h_{i}^{*}(v|w)=P\left(\left\{ \mathbf{x}:s_{X,i}^{*}(\mathbf{x})=v\right\} \mid\mathbf{s}_{C,X,i}^{*}(\mathbf{c})\right)$.

Potential outcomes under dynamic and stochastic interventions are
identified under the same no unmeasured confounding and positivity assumptions
as deterministic interventions. The conditional support of $V^{*}$
must be included in the conditional support of $V$ in order for the
intervention to be supported by the data and the positivity assumption
to hold. For dynamic interventions, identification follows (2) and
(3); the only difference is that $x_{i}^{*}$ is determined by $w_i$
rather than being specified directly. Under a stochastic intervention,
$X_{i}^{*}$ and $V_{i}^{*}$ are random and potential outcomes are
identified as follows:
\begin{align*}
E\left[\bar{Y}_{n}^{*}\right] & =\frac{1}{n}\sum_{i=1}^{n}E[Y_{i}^{*}]\\
 & =\frac{1}{n}\sum_{i=1}^{n}\sum_{\mathbf{c}}E[Y|V_{i}^{*},W_{i}=s_{C,i}(\mathbf{c})]p_{C}(\mathbf{c})\\
 & =\frac{1}{n}\sum_{i=1}^{n}\sum_{w,v}E[Y_{i}|V_{i}=v,W_{i}=w]h_{i}^{*}(v,w)\\
 & =\frac{1}{n}\sum_{i=1}^{n}\sum_{w,v}m(v,w)h_{i}^{*}(v,w).
\end{align*}
This generalizes from Equation (3) which reflects the fact that, for a static intervention, $h_i^*(v,w)$ only puts mass on the single value $v=v_i^*$.
Letting $\psi_{n}$ denote this observed data functional, a sample
average influence function is given by $D_{n}(\mathbf{o})$ in equation
(5). The only piece of this IF that depends on the nature of the intervention
is the form of $h^{*}$.
The TMLE of $\psi_{n}$ is computed according to the steps outlined
in Section \ref{sec:methods};
the fact that $\mathbf{X}^{*}$ and $V^{*}$ may be random does not
affect the estimation algorithm. (Additional details and examples
can be found in \citealp{sofryginTechreport}.) Theorem
1 is agnostic about whether or not $h^{*}$ is degenerate; therefore
the same asymptotic results hold for the estimation of potential outcomes
under static, dynamic, and stochastic interventions.

\subsection{Peer effects}

Define $Y_{i}^{0}$ to be the outcome variable measured at a time
previous to the primary outcome measurement $Y_{j}$. Peer effects
are the class of causal effects of $Y_{j}^{0}$ on $Y_{i}$ for $A_{ij}=1$:
the effects of alters' outcomes on the subsequent outcome of
an ego. If we let $X_j=Y_{i}^{0}$, we can use the framework above to define, identify, and estimate static, dynamic, and stochastic peer effects. For example,  

In order to maintain the identifying assumptions A2b and A3b, the
time elapsed between $Y^{0}$ and $Y$ must permit transmission only
between nodes and their immediate alters. Otherwise, if the outcome
could have spread contagiously more broadly, there will be more dependence
present than our methods can account for, and also possible confounding
of the effect of $Y_{i}^{0}$ on $Y_{j}$ for $A_{ij}=1$ due to mutual
connections.

\subsection{Interventions on network structure}

As a special case of interventions on $s_{X}(\cdot)$ or on $s_{C,X}(\cdot)$
we consider interventions determined by changes to the network itself,
i.e. interventions that add, remove, or relocate ties in the network. Consider an intervention that modifies $s_{X}$ by replacing the observed
adjacency matrix $\mathbf{A}$ with a user-specified adjacency matrix
$\mathbf{A}^{*}$. This intervention replaces $s_{X,i}(\mathbf{X})$
with $s_{X,i}^{\mathbf{A^{*}}}(\mathbf{X})\equiv s_{X,i}\left(\left\{ X_{j}:A_{ij}^{*}=1\right\} \right)$.
The intervention SEM differs from the data-generating SEM only in
that $Y_{i}$ depends on the counterfactual treatments for the individuals
with whom $i$ shares ties in the intervention adjacency matrix $\mathbf{A}^{*}$.
Similarly, an intervention on $s_{C,X}(\cdot)$ would result in a dynamic intervention where $X_{i}^{*}$
depending on the covariate values for the individuals with whom $i$
shares ties in the intervention adjacency matrix $\mathbf{A}^{*}$.

Interventions on summary features of the adjacency matrix can 
be operationalized as stochastic interventions. Instead of replacing $\mathbf{A}$
with a user-specified $\mathbf{A}^{*}$, an intervention on features of the network
structure might replace $\mathbf{A}$ with a random draw from a class $\mathcal{A}^{*}$
of $n\times n$ adjacency matrices that share the intervention features,
stochastically according to some probability distribution $p_{\mathbf{A}^{*}}$
over $\mathcal{A}^{*}$. For example, we might be interested in an intervention on $s_X$ 
that constrains the degree distribution of the network, e.g. fixing
the maximum degree to be smaller than some $D$. We might specify
$p_{\mathbf{A}^{*}}(A)=\frac{1}{\left|\mathcal{A}^{*}\right|}I\left\{ A\in\mathcal{A}^{*}\right\} $,
giving equal weight to each realization in the class $\mathcal{A}^{*}$.
This kind of intervention sets $V_{i}$ to $v$ with
probability 
$P\left[\mathbf{X}\in\left\{ \mathbf{x}:s_{X,i}^{\mathbf{A}^{*}}(\mathbf{x})=v\right\} \right]$.
Note that defining a feasible intervention with respect to changes
to the adjacency matrix is possible due to assumption (A6). This
is a strong assumption; if network structure can affect $\mathbf{Y}$
via mechanisms of interest that do not operate through $s_{X}(\cdot)$
then estimating these effects is more challenging \citep{ogburn2014causal,toulis2018propensity}.

As with the stochastic interventions discussed in the previous section,
positivity is a crucial assumption for identifying interventions on
$\mathbf{A}$: the support of $V^{*}$ must be the same as the support
of $V$. If replacing $\mathbf{A}$ with $\mathbf{A}^{*}$ 
assigns to unit $i$ a value of $V$ that not observed in the real
data for a unit in the same $W$ stratum as $i$, then the
effect of the intervention that replaces $\mathbf{A}$ with $\mathbf{A}^{*}$
is not identified for unit $i$. In general it may be possible to
identify interventions on local but not global features of network
structure. Examples of local features of network structure include
the degree of subject $i$ and local clustering around subject $i$:
they depend on $\mathbf{A}$ only through subject $i$ and subject
$i$'s immediate contacts. A local clustering coefficient for node
$i$ can be defined as the proportion of potential triangles that
include $i$ as one vertex and that are completed, or the number of
pairs of neighbors of $i$ who are connected divided by the total
number of pairs of neighbors of $i$ \cite{newman2009networks}. This
measure of triangle completion captures the extent to which ``the
friend of my friend is also my friend'': triangle completion is high
whenever two subjects who share a mutual contact are more likely to
themselves share a tie than are two subjects chosen at random from
the network. Positivity could hold if, within each level of $W$,
subjects were observed to have a wide range of degrees and of triangle
completion among their contacts. In contrast with degree and local
clustering, network centrality is a node-specific attribute that nevertheless
depends on the entire network structure. It captures the intuitive
notion that some nodes are central and some nodes are fringe in any
given network. It can be measured in many different ways, based, for example, on the number of network paths that intersect node $i$, on the probability that a random walk on the network will intersect node $i$, or on the mean distance between node $i$ and the other nodes in network (see Chapter 7 of \cite{newman2009networks} for
a comprehensive discussion of these and other centrality measures).
Centrality is given by a univariate measure for each node in a network,
but each node's measure depends crucially on the entire graph. In
reality it is not generally possible to intervene on centrality without
altering the entire adjacency matrix $\mathbf{A}$, and the positivity
assumption is unlikely to hold.


\subsection{Too many friends, too much influence} \label{sec:hubs}

The conditions of Theorem 1 will be
violated for any asymptotic regime in which the degree of one
or more nodes grows at a rate equal to or faster than $\sqrt n$. This is problematic because social networks frequently have
a small number of ``hubs''--that is, nodes with very high degree
\citep{newman2009networks}, 
When a small number of individuals wield influence
over a significant portion of the rest of the population, two problems
arise for statistical inference. First, the number of hubs may stay small as $n$ increases. If the hubs are systematically different from the rest of the population,
then a fixed or slowly growing number of hubs would not allow for
consistent inference about this distinct subpopulation. Second, and
more importantly, the sweeping influence of hubs creates dependence
among all of the influenced nodes that undermines inference. Our methods
rely on the independence of $Y_{i}$ and $Y_{j}$ whenever nodes $i$
and $j$ do not share a tie or a mutual alter. When hubs are present,
a significant proportion of nodes will share a connection to one of
these hubs, undermining our methods.

We can recover valid inference using our methods if we condition on
the hubs, treating them as features of the background network environment
rather than as observations. This results in different causal effects
or statistical estimands, as all of our inference is conditional on
the identity and characteristics of the hubs. Imagine a social network
comprised of the residents of a city in which a cultural or political
leader is connected to almost all of the other nodes. It may be impossible
to disentangle the influence of this leader, which affects every other
node, from other processes simultaneously occurring among the other
residents of the city. It will certainly be impossible to statistically
learn about the hub, as the sample size for the hub subgroup is 1.
But it may make sense to consider the hub as a feature of the city
rather than a member of the network. We could then learn about other
processes occurring among the other residents of the city, conditional
on the behavior and characteristics of the leader. For example, we
could evaluate the effect of a public health initiative encouraging
residents to talk to their friends about the importance of exercise,
but we could not evaluate a similar program targeting the leader's
communication about exercise.  

Practically speaking, this implies that
the methods we have proposed are inappropriate for networks in which
the degree is large, compared to $n$, for one or more nodes. If many
nodes are connected to a significant fraction of other nodes, this
problem is intractable. However, if only a small number of nodes are
highly connected we can condition on them to recover approximately
valid inference using our methods for conditional estimands. There
is a theoretical tradeoff between the rate of convergence of our estimators
and the order of $K$ relative to $n$ that, in finite samples, becomes
a practical tradeoff between generality and variance. Increasing the
number of nodes classified as hubs will increase the rate of convergence by decreasing the
size of $K$ for the remaining, non-hub nodes (assuming that the
number of hubs remains small compared to $n$ so that the sample size
does not decrease significantly when we exclude hubs from the analysis).
On the other hand, classifying more nodes as hubs results in analyses
that are increasingly specific: conditioning on a single hub may preserve
generalizability to other networks (similar cities with similar leaders),
but conditioning on many hubs is likely to limit the generalizability
of the resulting inference.

\section{Simulations}\label{sec:Simulations}

We conducted a simulation study that evaluated the finite sample and
asymptotic behavior of the TMLE procedure described in Section \ref{sub:Estimation}.
We generated social networks of size $n=500$, $n=1,000$, and $n=10,000$ according to the preferential attachment
model \citep{barabasi1999emergence}, where the node degree (number
of alters) distribution followed a power law with $\alpha=0.5$. We generated data with two different types of dependence: first with dependence due to direct
transmission only, and second with both latent variable dependence
and dependence due to direct transmission. Details of the simulations, along with results for networks generated under the small world model \citep{watts1998collective}, are in the Appendix.  

Our simulations mimicked a hypothetical study designed to increase
the level of physical activity in a population comprised of members
of a social network. For each community member indexed by $i=1,\ldots,n$,
the study collected data on $i$'s baseline covariates, denoted $C_{i}$,
which included the indicator of being physically active, denoted $PA_{i}$
and the network of alters--or friends--on each subject, $F_{i}$. The exposure
or treatment, $X_{i}$, was assigned randomly to 25\% of the community.
For example, one can imagine a study where treated individuals received
various economic incentives to attend a local gym. The outcome $Y_{i}$
was a binary indicator of maintaining gym membership for a pre-determined
follow-up period. We 
estimated the average of the mean counterfactual outcomes $E\left[\bar{Y}^{*}_{n}\right]$
under various hypothetical interventions $h^{*}$ on such a community.
First, we considered a stochastic intervention $h_{1}^{*}$ which
assigned each individual to treatment with a constant probability
of $0.35$; this differs from the observed allocation of treatment
to $25\%$ of the community members. We also considered a scenario
in which the economic incentive was resource constrained
and could only be allocated to up to 10\% of community members. We estimated the effects of various targeted approaches to allocating
the exposure. For example, we considered an intervention $h_{2}^{*}$
that targeted only the top 10\% most connected members of the community,
as such a targeted intervention would be expected to have a higher
impact on the overall average probability of maintaining gym membership
among the community, when compared to purely random assignment of
exposure to 10\% of the community. Another hypothetical intervention
$h_{3}^{*}$ assigned an additional physically active friend to individuals
with fewer than 10 friends. This is an intervention on the structure of the social network itself. Finally,
we estimated the combined effect of simultaneously implementing intervention
$h_{2}^{*}$ and the network-based intervention $h_{3}^{*}$ on the
same community. For simplicity, we report the expected outcome under each of these interventions;
causal effects defined as contrasts of these interventions can be
easily estimated using the same methods.

The results from the simulations with dependence due to direct transmission
are shown in the left panel of Figure \ref{fig:CIres.EY.prefattach}.
We estimated the marginal parameter $E\left[\bar{Y}^{*}_{n}\right]$
and compared three different estimators of the asymptotic variance
and the coverage of the corresponding confidence intervals. First,
we looked at the naive plug-in i.i.d.\ estimator (``\emph{i.i.d.\ Var}'')
for the variance of the influence curve which treated observations
as if they were i.i.d. Second, we used the plug-in variance estimator
based on the influence function which adjusted for the correlated
observations (``\emph{dependent IC Var}'') \citep{sofryginTechreport}. Finally, we used the parametric bootstrap
variance estimator (``\emph{bootstrap Var}'') described in Section \ref{sec:inference}. The results from the simulations with latent variable dependence are in the right panel of Figure \ref{fig:CIres.EY.prefattach}. We estimated
the conditional parameter $E\left[\bar{Y}^{*}_{n}\right]$ and compared
two plug in variance estimators based on the conditional influence
function $D_{n}^C$: one that assumes conditionally i.i.d outcomes (conditional
on $\mathbf{X}$ and $\mathbf{C}$), which would be true if all dependence
were due to direct transmission but is violated in the presence of
latent variable dependence (``\emph{i.i.d.\ Var}''), and one that does
not make this assumption (``\emph{dependent IC Var}''). In the Appendix we compare
histograms of the estimates to the predicted normal limiting distribution.

\begin{figure} [h]
\begin{centering}
\includegraphics[scale=0.32]{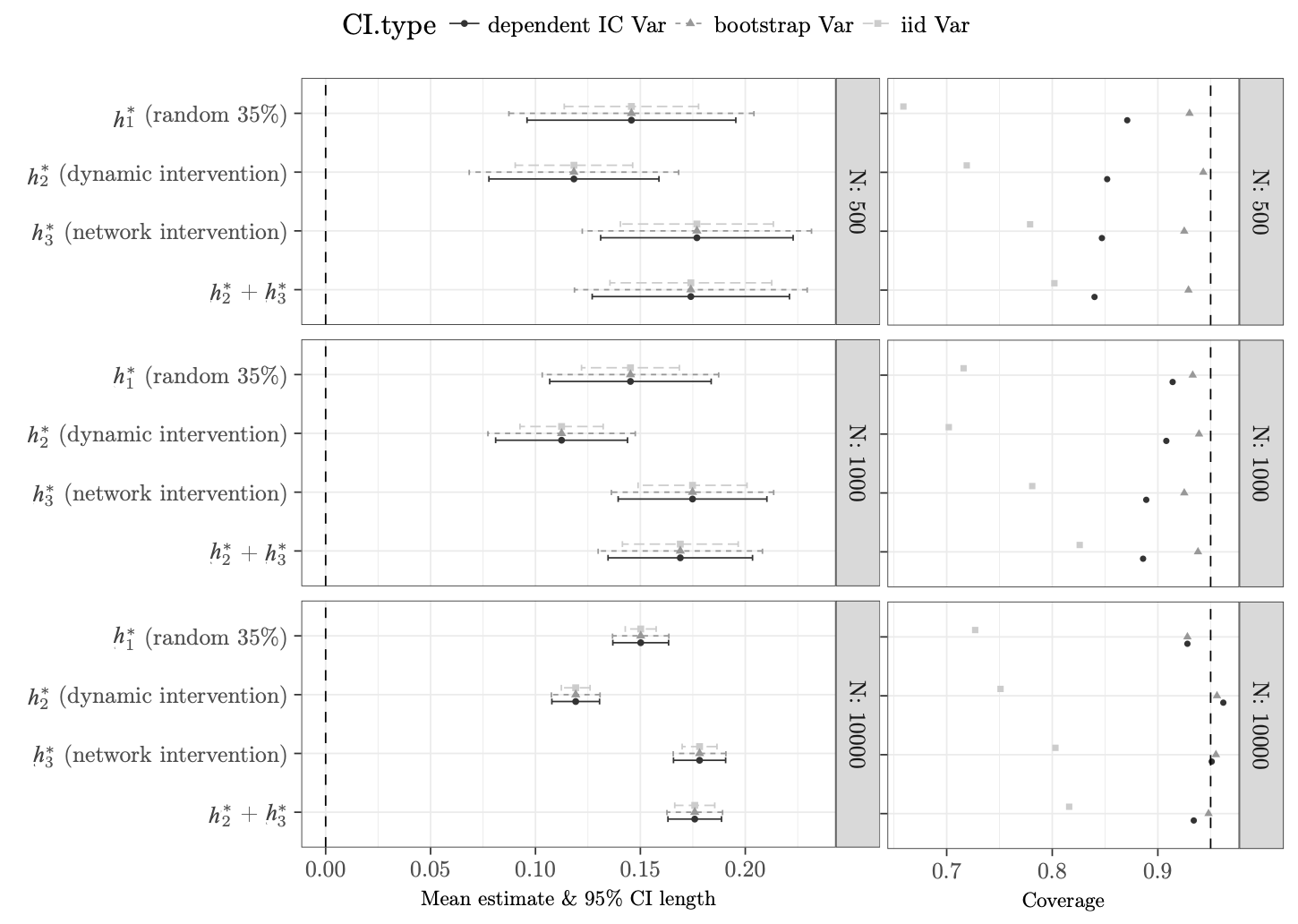} 
\includegraphics[scale=0.32]{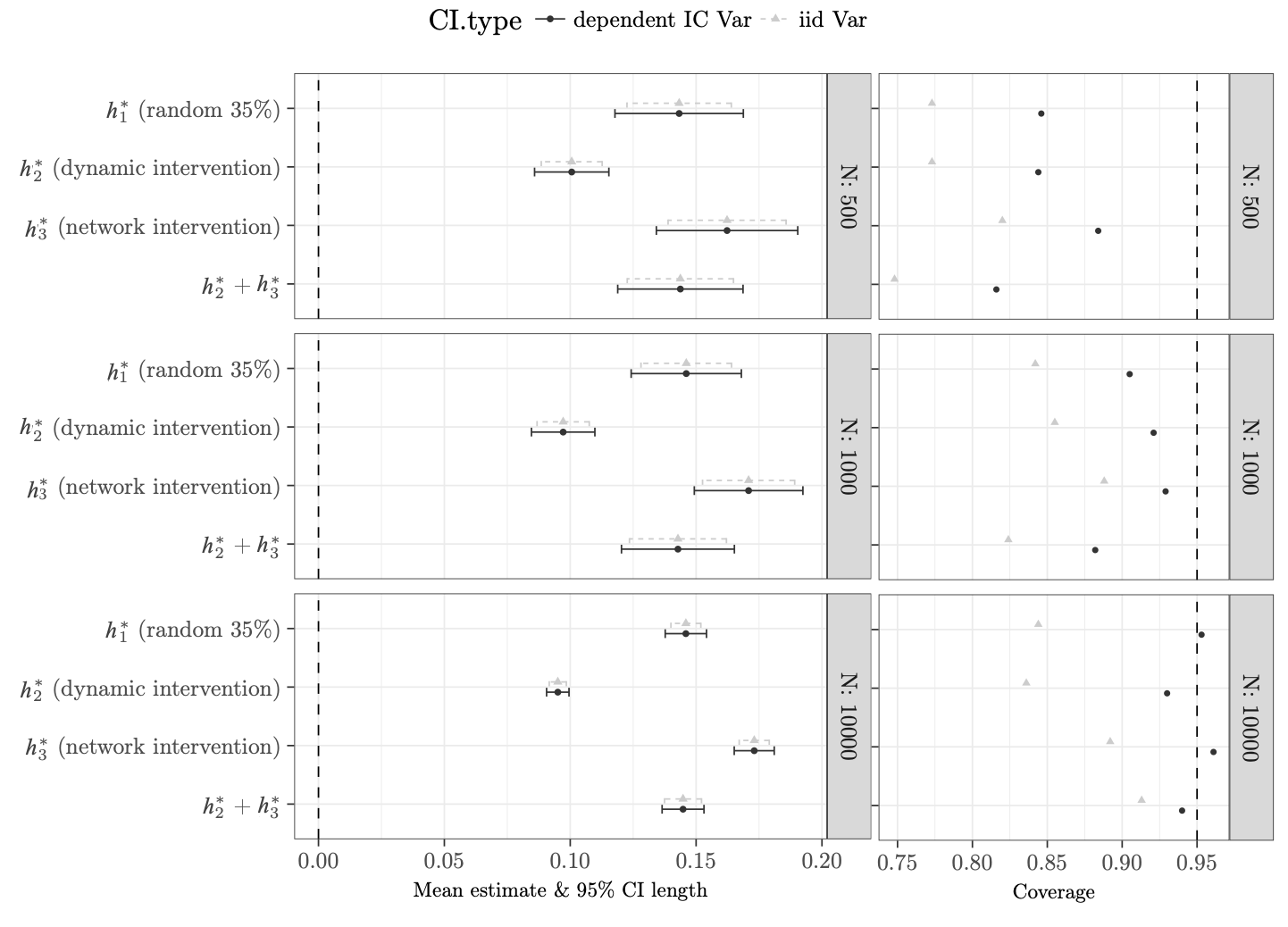} 
\par\end{centering}

\protect\protect\caption{Mean 95\% CI length and coverage for the
TMLE in preferential attachment network with dependence due to direct
transmission (left panel) and with latent variable dependence (right panel), by sample size, intervention and CI type. \label{fig:CIres.EY.prefattach}}
\end{figure}

One of the lessons of our simulation study is that by leveraging the
structure of the network it might be possible to achieve a larger
overall intervention effect on a population level \citep{harling2016leveraging}.
For example, the results in the left panel of Figure \ref{fig:CIres.EY.prefattach} show that by targeting the exposure
assignment to highly connected and physically active individuals,
intervention $h_{2}^{*}$ increases the mean probability of sustaining
gym membership compared to the similar level of un-targeted coverage
of the exposure. We also demonstrated the feasibility of estimating
effects of interventions on the observed network structure itself,
such as intervention $h_{3}^{*}$, which can be also combined with
economic incentives, as it was mimicked by our hypothetical intervention
$h_{2}^{*}+h_{3}^{*}$. These combined interventions could be particularly
useful in resource constrained environments, since they may result
in larger community level effects at the lower coverage of the exposure
assignment.

Results from simulations with dependence due to direct transmission
show that conducting inference while ignoring the nature of the dependence
in such datasets generally results in anticonservative variance estimates
and under-coverage of CIs, which can be as low as 50\% even for very
large sample sizes (``\emph{i.i.d.\ Var}'' in the left panel of Figure
\ref{fig:CIres.EY.prefattach}). The
CIs based on the dependent variance estimates (``\emph{dependent
IC Var}'') obtain nearly nominal
coverage of 95\% for large enough sample sizes, but can suffer in
smaller sample sizes due to lack of asymptotic normality and near-positivity
violations. Notably, the CIs based on the parametric bootstrap variance
estimates provide the most robust coverage for smaller sample sizes,
while attaining the nominal 95\% coverage in large sample sizes for
nearly all of the simulation scenarios (``\emph{bootstrap Var}''). The apparent robustness of the parametric bootstrap method for inference in small sample sizes,
even as low as $n=500$, was one of the surprising finding of this
simulation study. 
Similarly, in the simulations with latent variable
dependence the variance estimates that assume conditionally i.i.d.\
outcomes, i.e.\ that dependence may be due to direct transmission but
not to latent variables, are anti-conservative.

Code for all simulations is available in a github repository
(\verb!github.com/osofr/Ogburn_etal_simulations!).

\section{Data Analysis }
\label{sec:data}


We reanalyze an influential study that used the partially reconstructed social network of FHS study participants in order to study peer effects for obesity (\citealp{christakis2007spread}; hereafter CF). 
To assess peer influence for obesity using FHS data, CF fit longitudinal logistic regression models of each individual's obesity status at exam $k=2,3,4,5,6,7$  onto each of the individual's social contacts' obesity statuses at exam $k$ and $k-1$ with a separate entry into the model for each contact, controlling for individual covariates and for the node's own obesity status at exam $k-1$. They used generalized estimating equations to account for correlation within individual over time, but their model assumes independence across individuals. CF fit this model separately for ten different types of social connections, including siblings, spouses, and immediate neighbors, with estimates of the increased risk of obesity ranging from 27\% to 171\%, many of which were statistically significant.  In contrast with our approach, in which each ego is treated as a single (possibly dependent) observation, the pairwise approach that treats each network tie as an independent observation can result in incoherent models for the full network  \citep{lyons2011spread,ogburn2014vaccines}. Furthermore, \cite{lee2020network} found evidence of significant network dependence across observations, suggesting that even if the model were coherent the analysis is invalid due to unaccounted statistical dependence.  However, until now no method has been available to reanalyze these data taking into account the network structure and corresponding causal and statistical dependence.  

We reanalyzed data from the first two exams, using all ten types of social connections simultaneously ($n=3766$). The full $\mathbf{R}$
code for this analysis is available in a github repository
 (\verb!github.com/osofr/Ogburn_etal_simulations!);
 public versions of FHS data through 2008 are available from the dbGaP database. 
 Instead of specifying pairwise models and treating each pair (i.e.\ each network tie) as an independent observation, our methods account for the entire social network structure and allow for considerable causal and statistical dependence among subjects. 
For each subject $i$ we specified $m(V_i,W_i)$ to be the regression model used in CF (2007), but with proportion of obese alters replacing the indicator that a single friend is obese at each visit.  That is, we specified that the expected probability of obesity for subject $i$ at visit $2$ is a function of the proportion of $i$'s alters who were obese at visit 2 (this is the exposure of interest), subject $i$'s obesity at visit 1, the proportion of $i$'s alters who were obese at visit 1, and subject $i$-specific covariates age, sex, and education.  CF argue that controlling for alters' obesity status at visit 1 controls for confounding due to homophily.  It is more likely that confounding due to homophily cannot be controlled using these data \citep{shalizi2011homophily, cohen2008obesity, noel2011unfriending} and we do not purport to be estimating a true, unconfounded causal effect.  However, under CF's assumption of unconfoundedness, we can estimate the expected proportion of subjects who would be obese at visit 2 under various hypothetical interventions on each subject's alters' obesity statuses.  

The pairwise parameter that CF estimated is not well-defined in a model that accounts for more than one tie simultaneously. Instead we estimated the expected probability of obesity at visit 2 under a hypothetical intervention to increase the number of each subject's obese alters by 1; this dynamic intervention is well-defined under the assumption that a node's alters' obesity has an effect only through a parsimonious summary measure $s_X$. This intervention is similar to CF's pairwise parameter in that it estimates the effect of a single alter's change in obesity status. The observed empirical probability of obesity at visit 2 was 0.137. The predicted outcome under intervention was identical up to three decimal places with a 95\% parametric bootstrap confidence interval of $(0.127,  0.147)$.  
We re-ran the CF analysis using only outcome data from the second visit in order to more closely track our own analysis. We then used the single time point CF logistic regression to predict the outcome under an intervention to set the alter in each ego-alter pair to have obesity status $1$; this corresponds more closely (though not directly) to our causal estimand. Compared to the observed empirical probability of obesity of $0.137$, the predicted outcome under intervention was $0.147$ with a 95\% bootstrap confidence interval of $(0.138, 0.156)$. Like the original CF analyses, this seems to support the claim of peer effects for obesity. The non-null point estimate could be due in part to spurious associations due to dependence \citep{lee2020network}, and the slightly shorter confidence interval compared to our analysis could be due to underestimated variance. 

In summary, our analysis is consistent with the hypothesis that the significant results of CF are spurious, due to dependence and/or model misspecification rather than true associations or effects.
We caution against interpreting our estimates as true causal effects, both because of unobserved confounding in the FHS data and because the exposure was measured at the same time as the outcome.  However, this is still an instructive comparison between our methods and the naive methods that are currently in common use. Accounting for the interdependence of the subjects in the FHS data 
undermines the findings of strong contagion effects for obesity.  

\section{Conclusion} \label{sec:Conclusion}

We proposed new methods that allow for causal and statistical inference
using data from a single interconnected
social network, with causal and statistical dependence informed by network ties. In contrast to existing methods, our methods do not require
randomization of an exogenous treatment and they have proven performance
under asymptotic regimes in which the number of network ties grows
(slowly) with sample size.  In the absence of appropriate methods for assessing peer effects researchers have routinely relied on naive methods developed for independent units, and our analysis of peer effects for obesity in the Framingham Heart Study illustrates the dangers of that approach and the importance of new methods like ours.  


\section*{Acknowledgements}
The authors are grateful to Caleb Miles, Eric Tchetgen Tchetgen and Victor De Gruttola for helpful comments.  Elizabeth L. Ogburn was supported by ONR grant N000141512343 and N000141812760.   Oleg Sofrygin and Mark van der Laan were supported by NIH grant R01 AI074345-07.  
\newpage
\begin{center}
\centering
{
\textbf{\Huge{Appendix}}
}
\end{center}

\part{Proof of Theorem 1}
\section{Regularity conditions}

For a real-valued function $\mathbf{c}\mapsto f(\mathbf{c})$, let
the $L^{2}(P)$-norm of $f(\mathbf{c})$ be denoted by $\norm{f}=E[f(\mathbf{C})^{2}]^{1/2}$.
Define \textbf{$\mathcal{M}_{m}$ }and \textbf{$\mathcal{M}_{\tilde{h}}$
}as the classes of possible functions that can be used for estimating
the two nuisance parameters $m$ and $\tilde{h}\equiv\bar{h}_{x^{*}}/\bar{h}$,
respectively. Note that a model for $g$ plus the empirical distribution
of covariates $\mathbf{C}$ determines $\tilde{h}$. Equivalent assumptions
could be stated in terms of $g$ instead of $\tilde{h}$, but we focus
on $\tilde{h}$ because that is the functional of $g$ and $\mathbf{C}$
that we model in our estimating procedure. Assume that the TMLE update
$\hat{m}_{\hat{\epsilon}}\in{\cal M}_{m}$ with probability 1 and
assume that $\hat{\bar{h}}_{x^{*}}/\hat{\bar{h}}\in{\cal M}_{\tilde{h}}$
with probability 1. Finally, define the following dissimilarity measure
on the cartesian product of $\mathcal{F}\equiv{\cal M}_{m}\times{\cal M}_{\tilde{h}}$:
\[
d\left(\left(h,m\right),\left(\tilde{h},\tilde{m}\right)\right)=\max\left(\sup_{v\in{\cal V}}\mid h-\tilde{h}\mid(v),\sup_{v\in{\cal V}}\mid m-\tilde{m}\mid(v)\right).
\]

The following are the regularity conditions required for Theorem 1,
i.e. for asymptotic normality of the TMLE $\hat{\psi}^{*}$. 
\begin{description}
\item [{Uniform consistency:}] Assume that 
\[
d\left(\left(\hat{\bar{h}}_{x^{*}}/\hat{\bar{h}},\hat{m}_{\hat{\epsilon}}\right),\left(\bar{h}_{x^{*}}/\bar{h},m\right)\right)\rightarrow0
\]
in probability as $n\rightarrow\infty$. Note that this assumption
is only needed for proving the asymptotic equicontinuity of our process;
it is not needed for proofs of relevant convergence rates for the
second order terms. 
\item [{Bounded entropy integral:}] Assume that there exists some $\eta>0$,
so that $\int_{0}^{\eta}\sqrt{\log\left(N(\epsilon,{\cal F},d)\right)}d\epsilon<\infty$,
where $N(\epsilon,{\cal F},d)$ is the number of balls of size $\epsilon$
w.r.t. metric $d$ needed to cover ${\cal F}$. 
\item [{Universal bound:}] Assume $\sup_{f\in{\cal F},\mathbf{O}}\mid f\mid(\mathbf{O})<\infty$,
where the supremum of $\mathbf{O}$ is over a set that contains $\mathbf{O}$
with probability one. This assumption will typically be a consequence
of the choosing a specific function class ${\cal F}$ that satisfies
the above entropy condition. 
\item [{Positivity:}] Assume 
\[
\sup_{v\in{\cal V},w\in{\cal W}}\frac{\bar{h}_{x^{*}}(v,w)}{\bar{h}(v,w)}<\infty.
\]

\item [{Consistency and rates for estimators of nuisance parameters:}] Assume
that $\norm{\hat{m}-m}\norm{\hat{\bar{h}}-\bar{h}}=o_{P}\left(\left(C_{n}\right)^{-1/2}\right)$.
Note that this rate is achievable if, for example, estimation of $\bar{h}$
relies on some pre-specified parametric model, or if both $\bar{h}$
and $m$ are estimated at rate $C_{n}^{-1/4}$.
\item [{Rate of the second order term:}] Assume that 
\[
R_{n1}\equiv-\int_{v,w}\left\{ \left(\frac{\hat{\bar{h}}_{x^{*}}}{\hat{\bar{h}}}-\frac{\bar{h}_{x^{*}}}{\bar{h}}\right)(\hat{m}_{\hat{\epsilon}}-m)(v,w)\bar{h}(v,w)d\mu(v,w)\right\} =o_{P}\left(1/\sqrt{C_{n}}\right).
\]
Note that this condition is provided here purely for the sake of completeness,
since it will satisfied based on the previously assumed rates of convergence
for $\norm{\hat{m}-m}\norm{\hat{\bar{h}}-\bar{h}}$. This follows
from the fact that the parametric TMLE update step $\hat{m}_{\hat{\epsilon}}$
of $\hat{m}$ will have a negligible effect on the rate of convergence
of the initial estimator $\hat{m}$, that is, $\hat{m}_{\hat{\epsilon}}$
will converge at ``nearly'' the same rate as $\hat{m}$. 
\item [{Limited connectivity and limited dependence of $\mathbf{Y}$,$\mathbf{X}$ and $\mathbf{C}$:}] Let
$K_{max,n}=max_{i}\{K_{i}\}$ for a fixed network with $n$ nodes.
Assume that $K_{max,n}^{2}/n$ converges to $0$ in probability as
$n\rightarrow\infty$. 
\end{description}
A key condition is \textit{consistency and rates for estimators of
nuisance parameters}. This condition will be satisfied, for example,
if both models converge to the truth at rate $C_{n}^{1/4}$.
It can in fact be weakened, but for a more general discussion and
the corresponding technical conditions we refer to the Appendix of
\citet{vanderlaan2012}. With the exception of the rates of convergence,
the more general conditions for asymptotic normality of the TMLE presented
in that paper apply to our setting as well.
{max,n}

\section{Overview of the proof of Theorem 1 }

We want to show that $\sqrt{C_{n}}(\hat{\psi}-\psi)$ converges in
law to a Normal limit as $n$ goes to infinity for some rate $\sqrt{C_{n}}$
such that $\sqrt{n/\left(K_{max,n}\right)^{2}}\leq\sqrt{C}_{n}\leq\sqrt{n}$,
where the rate $\sqrt{C_{n}}$ is the order of the variance of the
sum of the first-order linear approximation of $(\hat{\psi}-\psi)$.

Broadly, the proof has two parts: First, we require that the second
order terms in the expansion of $\hat{\psi}-\psi$ are stochastically
less than $1/\sqrt{C_{n}}$, that is that 
\begin{eqnarray*}
\hat{\psi_{n}}-\psi & = & \frac{1}{n}\sum_{i=1}^{n}\left\{ f_{i}(\mathbf{O})-E[f_{i}(\mathbf{O})]\right\} +o_{p}\left(1/\sqrt{C_{n}}\right),\\
\end{eqnarray*}
where $f_{i}(\mathbf{O})$ is the contribution of the $i$th observation
to the estimator. 
Then proving asymptotic normality of the TMLE amounts to the asymptotic
analysis of the sum $\frac{1}{n}\sum_{i=1}^{n}\left\{ f_{i}(\mathbf{O})-E[f_{i}(\mathbf{O})]\right\} $,
and the second part of the proof establishes that the first order
terms converge to a normal distribution when scaled by $\sqrt{C_{n}}$,
that is that $\sqrt{C_{n}}\frac{1}{n}\sum_{i=1}^{n}\left\{ f_{i}(\mathbf{O})-E[f_{i}(\mathbf{O})]\right\} \rightarrow_{d}N(0,\sigma^{2})$ for some finite $\sigma^{2}$.

The proof that the second order terms are stochastically less than
$1/\sqrt{C_{n}}$ is an extension of the empirical process theory
of \citet{van1996weak} and follows the same format as the proof in
\citet{vanderlaan2012}. Indeed, the proof offered by \citet{vanderlaan2012}
holds immediately after replacing the rate or scaling factor $\sqrt{n}$
with $\sqrt{C_{n}}$ throughout. Only one step in the \citet{vanderlaan2012}
proof relies on the network structure, which is the major difference
between the setting in that paper, where the number of network connections
is fixed and bounded as $n$ goes to infinity, and the present setting:
the proof requires bounding the Orlicz norms of several empirical
processes corresponding to components of the influence function for
$\psi$, and a key step is bounding the expectation of $E\left[\left|X_{n}(f)\right|^{p}\right]$
, where $X_{n}(f)$ is the stochastic process that describes the difference
between the empirical (indexed by $n$) and the true distribution
functions of a component of the influence function for $\psi$. This
step relies on a combinatorial argument about nature of overlapping
friend groups in the underlying network, and the argument for the
case of growing $K_{i}$ is subsumed by the argument for fixed $K$
in \citet{vanderlaan2012}.

The proof that the first order terms converge to a normal distribution
requires a central limit theorem for dependent data with growing and
possibly irregularly sized dependency neighborhoods, where a dependency
neighborhood for unit $i$ is a collection of observations on which
the observations for unit $i$ may be dependent. We prove such a CLT
in Lemmas 1 and 2. In the next section we use the CLT for growing
and irregular dependency neighborhoods, along with an orthogonal decomposition
of the first order terms, to prove the remainder of Theorem 1.

\section{Central limit theorem for first order terms}

Proving asymptotic normality of the TMLE amounts to the asymptotic
analysis of the sum $\frac{1}{n}\sum_{i=1}^{n}\left\{ f_{i}(\mathbf{O})-E[f_{i}(\mathbf{O})]\right\} $.
As a start, decompose $\sum_{i=1}^{n}\left\{ f_{i}(\mathbf{O})-E[f_{i}(\mathbf{O})]\right\} $
into a sum of three orthogonal components: 
\begin{eqnarray*}
f_{\mathbf{Y},i}(\mathbf{Y},\mathbf{X},\mathbf{C}) & = & f_{i}(\mathbf{O})-E\left[f_{i}(\mathbf{O})\mid\mathbf{X},\mathbf{C}\right],\\
f_{\mathbf{X},i}(\mathbf{X},\mathbf{C}) & = & E[f_{i}(\mathbf{O})\mid\mathbf{X},\mathbf{C}]-E[f_{i}(\mathbf{O})\mid\mathbf{C}],\mbox{ and}\\
f_{\mathbf{C},i}(\mathbf{C}) & = & E[f_{i}(\mathbf{O})\mid\mathbf{C}]-E[f_{i}(\mathbf{O})].
\end{eqnarray*}
Note that 
\[
f_{i}(\mathbf{O})-E[f_{i}(\mathbf{O})]=f_{\mathbf{Y},i}(\mathbf{Y},\mathbf{X},\mathbf{C})+f_{\mathbf{X},i}(\mathbf{X},\mathbf{C})+f_{\mathbf{C},i}(\mathbf{C})
\]
and with slight abuse of notation we will also write $f_{\mathbf{Y},i}(\mathbf{O})$,
$f_{\mathbf{X},i}(\mathbf{O})$ and $f_{\mathbf{C},i}(\mathbf{O})$.
Let $f_{\mathbf{Y}}(\mathbf{O})=\sum_{i=1}^{n}f_{\mathbf{Y},i}(\mathbf{O})$,
$f_{\mathbf{X}}(\mathbf{O})=\sum_{i=1}^{n}f_{\mathbf{X},i}(\mathbf{O})$
and $f_{\mathbf{C}}(\mathbf{O})=\sum_{i=1}^{n}f_{\mathbf{C},i}(\mathbf{O})$.
For $i=1,\ldots,n$, let

\begin{eqnarray*}
Z{}_{Y,i} & = & \frac{f_{\mathbf{Y},i}(\mathbf{Y},\mathbf{X},\mathbf{C})}{\sqrt{Var(\sum_{i=1}^{n}f_{\mathbf{Y},i}(\mathbf{Y},\mathbf{X},\mathbf{C}))}}\\
Z{}_{X,i} & = & \frac{f_{\mathbf{X},i}(\mathbf{X},\mathbf{C})}{\sqrt{Var(\sum_{i=1}^{n}f_{\mathbf{X},i}(\mathbf{X},\mathbf{C}))}}\\
Z{}_{C,i} & = & \frac{f_{\mathbf{C},i}(\mathbf{C})}{\sqrt{Var(\sum_{i=1}^{n}f_{\mathbf{C},i}(\mathbf{C}))}}.
\end{eqnarray*}
and

\begin{eqnarray*}
Z'_{Y,i} & = & \frac{f_{\mathbf{Y},i}(\mathbf{Y},\mathbf{X},\mathbf{C})\left|(\mathbf{X},\mathbf{C})\right.}{\sqrt{Var(\sum_{i=1}^{n}f_{\mathbf{Y},i}(\mathbf{Y},\mathbf{X},\mathbf{C})\left|(\mathbf{X},\mathbf{C})\right.)}}\\
Z'_{X,i} & = & \frac{f_{\mathbf{X},i}(\mathbf{X},\mathbf{C})\left|\mathbf{C}\right.}{\sqrt{Var(\sum_{i=1}^{n}f_{\mathbf{X},i}(\mathbf{X},\mathbf{C})\left|\mathbf{C}\right.)}}\\
\end{eqnarray*}
We use the prime to denote conditional random variables: $Z'_{Y,i}$
conditions $f_{\mathbf{Y},i}(\mathbf{O})$ on $(\mathbf{X},\mathbf{C})$
and rescales it by the standard error of $f_{\mathbf{Y}}(\mathbf{O})|\left(\mathbf{X},\mathbf{C}\right)$.
Similarly, $Z'_{X,i}$ conditions $f_{\mathbf{X},i}(\mathbf{O})$
on $\mathbf{C}$ and rescales it by the standard error of $f_{\mathbf{X}}(\mathbf{O})|\mathbf{C}$.
Let 
\begin{align*}
 & \sigma_{nY}^{2}(\mathbf{x},\mathbf{\mathbf{c}})=Var\left(\sum_{i=1}^{n}f_{\mathbf{Y},i}(\mathbf{Y},\mathbf{x},\mathbf{c})\left|(\mathbf{X}=\mathbf{x},\mathbf{C}=\mathbf{\mathbf{c}})\right.\right)\\
 & \sigma_{nY}^{2}=E_{P_{\mathbf{X},\mathbf{C}}}\left[\sigma_{nY}^{2}(\mathbf{X},\mathbf{\mathbf{C}})\right],
\end{align*}
\begin{align*}
 & \sigma_{nX}^{2}(\mathbf{\mathbf{c}})=Var\left(\sum_{i=1}^{n}f_{\mathbf{X},i}(\mathbf{X},\mathbf{c})\left|\mathbf{C}=\mathbf{c}\right.\right)\\
 & \sigma_{nX}^{2}=E_{P_{\mathbf{C}}}\left[\sigma_{nX}^{2}(\mathbf{\mathbf{C}})\right],
\end{align*}
and 
\[
\sigma_{nC}^{2}=Var\left(\sum_{i=1}^{n}f_{\mathbf{C},i}(\mathbf{C})\right).
\]
Note that by the law of total variance $\sigma_{nX}^{2}=Var(\sum_{i=1}^{n}f_{\mathbf{X},i}(\mathbf{X},\mathbf{C}))$
and $\sigma_{nY}^{2}=Var\left(\sum_{i=1}^{n}f_{\mathbf{Y},i}(\mathbf{Y},\mathbf{X},\mathbf{C})\right)$. 

Let $Z_{nY}'$ denote $\sum_{i=1}^{n}Z_{Y,i}'$, $Z_{nX}'$ denote
$\sum_{i=1}^{n}Z_{X,i}'$, $Z_{nY}$ denote $\sum_{i=1}^{n}Z_{Y,i}$,
$Z_{nX}$ denote $\sum_{i=1}^{n}Z_{X,i}$, and $Z_{nC}$ denote $\sum_{i=1}^{n}Z_{C,i}$.
We will establish convergence in distribution of each of the three
terms separately. Because $Z_{nY}^{'}$ and $Z_{nX}^{'}$ converge
to distributions that do not depend on their conditioning events,
conditional convergence in distribution implies convergence of $Z_{nY}$
and $Z_{nX}$ to the same limiting distributions. Since $f_{Y}(\mathbf{O}),$$f_{X}(\mathbf{O})$,
and $f_{C}(\mathbf{O})$ are orthogonal by construction, the variance
of the limiting distribution of their sum is the sum of their marginal
variances. If the three processes converge at the same rate the limiting
variance will be the sum of the variances of the three processes.
However, the three terms may converge at different rates, in which
case the limiting distribution of $\hat{\psi}-\psi$ will be given
by the limiting distribution of the term(s) with the slowest rate
of convergence.

In order to show that $Z_{nX}'$, $Z_{nY}'$, and $Z_{nC}$ all converge
in distribution to a $N(0,1)$ random variable, we can use three separate
applications of the central limit theorem given in Lemma 1, which
is based on Stein's method.

Stein's method \citep{stein1972bound} quantifies the error in approximating
a sample average with a normal distribution. (For an introduction
to Stein's method see \citealp{ross2011fundamentals}.) Stein's method
has been used to prove CLTs for dependent data with dependence structure
given by \textit{dependency neighborhoods} \citep{chen2004normal}:
the dependency neighborhood for observation $i$ is a set of indices
$D_{i}$ such that observation $i$ is independent of observation
$j$, for any $j\notin D_{i}$. Conditionally on $\mathbf{C}$, $f_{X,i}$
and $f_{X,j}$ are independent for any nodes $i$ and $j$ such that
$A_{ij}=0$ and there is no $k$ with $A_{ik}=A_{jk}=1$, that is
for any nodes that do not share a tie or have any mutual network contacts.
The same is true for $f_{Y,i}$ and $f_{Y,j}$ conditional on $\mathbf{X}$
and $\mathbf{C}$ and for $f_{C,i}$ and $f_{C,j}$. Thus the three
collections of random variables $Z_{X,1}',...,Z_{X,n}'$, $Z'_{Y,1},...,Z'_{Y,n}$,
and $Z_{C,1},...,Z_{C,n}$ each has a dependency neighborhood structure
with $D_{i}=i\cup\{j:A_{ij}=1\}\cup\{k:A_{jk}=1\mbox{ for }j:A_{ij}=1\}$,
that is the ``friends'' and ``friends of friends'' of node $i$.
Define the indicators $R(i,j)$ for any $(i,j)\in\{1,\ldots,n\}^{2}$
to be an indicator of dependence between $Z_{X,i}$ and $Z_{X,j}$,
$R(i,j)=1$ iff $j\in D_{i}$ or, equivalently, if $i\in D_{j}$.
For any $i\in\{1,\ldots,n\}$ the set $\{Z'_{X,j}:\left(R(i,j)=1,j\in\{1,\ldots,n\}\right)\}$
forms the dependency neighborhood of $Z'_{X,i}$ and the collection
$\{Z'_{X,j}:\left(R(i,j)=0,j\in\{1,\ldots,n\}\right)\}$ is independent
of $Z'_{X,i}$. The same logic applies to defining the dependency
neighborhoods for $Z'_{Y,1},...,Z'_{Y,n}$ conditional on $\mathbf{X}$
and $\mathbf{C}$, and for $Z_{C,1},...,Z_{C,n}$ based on (unconditional)
independence of each $f_{C,i}(\mathbf{O})$ and $f_{C,j}(\mathbf{O})$,
as determined by the network structure and the distributional assumptions
made for the baseline covariates $\mathbf{C}$.

Applied to $Z'_{nX}$ , Stein's method provides the following upper
bound

\begin{eqnarray*}
d(Z'_{nX},Z) & \leq & \sum_{i=1}^{n}\sum_{j,k\in D_{i}}E\left|Z'_{X,i}Z'_{X,j}Z'_{X,k}\right|\\
 &  & +\sqrt{\frac{2}{\pi}}\sqrt{Var\left(\sum_{i=1}^{n}\sum_{j\in D_{i}}Z'_{X,i}Z'_{X,j}\right)},
\end{eqnarray*}
where $Z\sim N(0,1)$ and $d(\cdot,\cdot)$ is the Wasserstein distance
metric \citep{vallender1974calculation}.

In order to show that $Z'_{nX}$ converges in distribution to $Z$,
we must show that the righthand side of the inequality converges to
zero as $n$ goes to infinity. We will first show that this convergence
holds when $K_{i}=\left|F_{i}\right|=K_{max,n}$ for all $i$, that
is when all nodes have the same number of ties. We will then show
that removing any tie from the network preserves an upper bound on
the righthand side of the inequality. This completes our proof that
for any network such that $K_{i}\leq K_{max,n}$ for all $i$ and
$\frac{K_{max}^{2}(n)}{n}$ converges to zero as $n$ goes to infinity,
$Z'_{nX}$ converges in distribution to a standard normal distribution.
The same argument applied to $Z{}_{nC}$ proves that it has a Normal
limiting distributions as well.\\

\begin{lemma}[Applying Stein's Method to the dependent sum]
\label{lem:Steins}\label{stein's lemma}
Consider a network of nodes given by adjacency matrix $A$. Let $U_{1},...,U_{n}$
be bounded mean-zero random variables with finite fourth moments and
with dependency neighborhoods $D_{i}=i\cup\{j:A_{ij}=1\}\cup\{k:A_{jk}=1\mbox{ for }j:A_{ij}=1\}$,  
and let $K_{i}$ be the degree of node $i$. Note that the size of each dependency neighborhood is bounded above by $K_{max,n}^2$. If $K_{i}=K_{max,n}$
for all $i$ and $K_{max,n}^{2}/n\rightarrow0$, then $\frac{\sum U_{i}}{\sqrt{var(\sum U_{i})}}\overset{D}{\rightarrow}N(0,1)$.
\end{lemma}

\begin{proof}[Proof of Lemma~\ref{lem:Steins}]
Let $U'_{i}=\frac{U_{i}}{\sqrt{var(\sum U_{i})}}$. Application of Stein's
method often involves defining the so-called ``Stein coupling''
$(W,W',G)$ \citep{fang2011thesis,fang2015rates}. Consider the following
sum of dependent variables $W=\sum_{i=1}^{n}U'_{i}$. Define a discrete
random variable $I$ distributed uniformly over $\{1,\ldots,n\}$
and define another random variable $W'=(W-\sum_{j=1}^{n}R(I,j)U'_{j})$.
Finally, define $G=-nU'_{I}$ and note that $(W,W',G)$ forms a Stein
coupling \citep{fang2011thesis,fang2015rates}. We also let $D=(W'-W)=-\sum_{j=1}^{N}R(I,j)U'{}_{j}$.
This Stein coupling allows us then to derive the upper bound 
\begin{align}
d(W,Z) & \leq\sum_{i=1}^{n}\sum_{j,k\in D_{i}}E\left|U'{}_{i}U'_{j}U'_{k}\right|+\sqrt{\frac{2}{\pi}}\sqrt{Var\left(\sum_{i=1}^{n}\sum_{j\in D_{i}}U'_{i}U'_{j}\right)},\label{eq:stein}
\end{align}
 as shown in \citet{ross2011fundamentals}. We will now show that,
for any network structure, 
\begin{eqnarray}
 &  & \sum_{i=1}^{n}\sum_{j,k\in D_{i}}E\left|U'{}_{i}U'{}_{j}U'_{k}\right|+\sqrt{\frac{2}{\pi}}\sqrt{Var\left(\sum_{i=1}^{n}\sum_{j\in D_{i}}U'_{i}U'_{j}\right)}\nonumber \\
 & = & O\left(\frac{\sum_{i,j,k}R(i,j)R(i,k)}{\left[\sum_{i,j}R(i,j)\right]^{3/2}}\right).\label{eq:to show}
\end{eqnarray}
The righthand side of the above equation is equal to $\sqrt{\frac{\left(K_{max,n}\right)^{2}}{n}}$
under the assumption of $K_{max,n}$ ties for each node $i=\{1,\ldots,n\}$.
By assumption, we also have that $\frac{K_{max,n}}{\sqrt{n}}$ converges
to zero as $n$ goes to infinity, and therefore if we can show equation
\eqref{eq:to show} we have proved that $\frac{\sum U_{i}}{\sqrt{var(\sum U_{i})}}\overset{D}{\rightarrow}N(0,1)$.

Consider the term
\[
\sum_{i=1}^{n}\sum_{j,k\in D_{i}}E\left|U'{}_{i}U'{}_{j}U'_{k}\right|=\dfrac{1}{var(\sum U_{i})^{3/2}}\sum_{i=1}^{n}E\left\{ \left|U_{i}\left(\sum_{j\in D_{i}}U_{k}\right)^{2}\right|\right\} .
\]
By the assumption of bounded 4th moments, $var(\sum U_{i})^{3/2}=O\left(\left[\sum_{i,j}R(i,j)\right]^{3/2}\right)$,
that is, $var(\sum U_{i})$ stabilizes to a constant when scaled by
$\sum_{i,j}R(i,j)$. Using the fact that each $|U_{i}|$ is bounded
we get 
\begin{eqnarray*}
 &  & \sum_{i=1}^{N}E\left\{ \left|U_{i}\left(\sum_{j\in D_{i}}U_{j}\right)^{2}\right|\right\} \\
 & \leq & M\sum_{i=1}^{n}\left\{ \sum_{j,k}R(i,j)R(i,k)\right\} \\
 & = & M\sum_{i,j,k}R(i,j)R(i,k),
\end{eqnarray*}
for some positive constant $M<\infty$. Combining the above expressions,
we get 
\[
\sum_{i=1}^{n}\sum_{j,k\in D_{i}}E\left|U'{}_{i}U'{}_{j}U'_{k}\right|=O\left(\frac{\sum_{i,j,k}R(i,j)R(i,k)}{\left[\sum_{i,j}R(i,j)\right]^{3/2}}\right).
\]

Now consider the second term: 
\begin{align*}
\sqrt{Var\left(\sum_{i=1}^{n}\sum_{j\in D_{i}}U'_{i}U'_{j}\right)} & =\frac{\sqrt{Var\left(\sum_{i=1}^{n}\sum_{j\in D_{i}}U{}_{i}U{}_{j}\right)}}{var(\sum U_{i})^{2}}.
\end{align*}
There are $\sum_{i,j}R(i,j)$ terms in $\sum_{i=1}^{n}\sum_{j\in D_{i}}U{}_{i}U{}_{j}$,
and the number of terms $U_{k}U_{l}$ with which $U_{i}U_{j}$ has
non-zero covariance is $\left|D_{i}\cup D_{j}\right|\leq\sum_{k}R(i,k)+\sum_{k}R(i,k)$,
so $Var\left(\sum_{i=1}^{n}\sum_{j\in D_{i}}U{}_{i}U{}_{j}\right)\leq M\sum_{i,j}R(i,j)\sum_{k}R(i,k)$
for some finite $M$. Therefore $Var\left(\sum_{i=1}^{n}\sum_{j\in D_{i}}U{}_{i}U{}_{j}\right)=O\left(\sum_{i,j,k}R(i,j)R(i,k)\right)$.
$Var(\sum U_{i})^{2}=O\left(\left[\sum_{i,j}R(i,j)\right]^{2}\right)$,
so the second term is of smaller order than the first term. Therefore
we have only to consider the first term and we have completed the
proof. \end{proof}

\begin{lemma}[Bound goes to zero when $K_{i}\leq K_{max,n}$ for
all $i$]
\label{lem:lowerbound} 
Convergence to zero of the righthand
side of Equation (\ref{eq:stein}) is preserved under the removal
of ties and holds as long as $K_{i}\leq K_{max,n}$ for all $i$
and $\frac{K_{max}^{2}(n)}{n}$ converges to zero as $n$ goes to
infinity. \end{lemma}

\begin{proof}[Proof of Lemma~\ref{lem:lowerbound}]
Consider a sequence of networks with $n$ going to infinity such that the righthand
side of Equation (\ref{eq:stein}) converges to $0$, i.e. 
\begin{eqnarray*}
\sum_{i=1}^{n}\sum_{j,k\in D_{i}}E\left|U'_{i}U'_{j}U'_{k}\right|+\sqrt{\frac{2}{\pi}}\sqrt{Var\left(\sum_{i=1}^{n}\sum_{j\in D_{i}}U'_{i}U'_{j}\right)} & \rightarrow & 0.
\end{eqnarray*}
Because the second term is of the same or smaller order than the first,
we only have to consider the first term. For this sequence of networks,
define $A_{n}=\sum_{i=1}^{n}\sum_{j,k\in D_{i}}E\left|U'_{i}U'_{j}U'_{k}\right|$
. Removing a single tie from the underlying network has the effect
of rendering independent some pairs that were previously dependent;
We now consider the effect of rendering a single dependent pair independent
but otherwise leaving the distributions of the random variables the
same. Suppose the pair rendered independent is $(l,m)$. Define a
new sequence of networks with $n$ going to infinity to be identical
to the previous sequence but with pair $(l,m)$ independent, and let
$A'_{n}$ be the first term in the righthand side of Equation (\ref{eq:stein})
for this new sequence. Then 
\begin{align*}
A'_{n} & =A_{n}-2\sum_{k\in D_{l}\cup D_{m}}E\left|U'_{l}U'_{m}U'_{k}\right|
\end{align*}
 which is bounded above by $A_{n}$.\end{proof}

This completes the proof that $Z_{nX}'$, $Z'_{nY}$, and $Z{}_{nC}$
have Normal limiting distributions.

\begin{lemma}[Conditional CLT implies marginal CLT]\label{lem:marginalCLT}$Z_{nX}'$
converges to Normal distribution after marginalizing over \textbf{$\mathbf{C}$
}(but conditioning on the network as captured by the adjacency matrix
\textbf{$\mathbf{A}$}) and $Z_{nY}'$ converges to Normal distribution
after marginalizing over $(\mathbf{X},\mathbf{C})$. That is, $Z_{nX}$
and $Z_{nY}$ both converge to Normal distributions.\end{lemma}

\begin{proof}[Proof of Lemma~\ref{lem:marginalCLT}] For illustration
consider $Z_{nX}'=\sum_{i=1}^{n}Z_{2,i}'$, where 
\[
Z_{X,i}'=\left(f_{\mathbf{X},i}(\mathbf{X},\mathbf{C})\left|\mathbf{C}\right.\right)/\sqrt{\sigma_{nX}^{2}(\mathbf{\mathbf{C}})}
\]
and note that the proof of the convergence of $Z_{nY}$ is nearly
identical. The conditional CLT results from Lemma \ref{lem:Steins}
show that 
\[
P\left[Z_{nX}'\leq x\left|\mathbf{C}=\mathbf{c}\right.\right]=P\left[\left(\sum_{i=1}^{N}\dfrac{f_{\mathbf{X},i}(\mathbf{X},\mathbf{c})}{\sqrt{\sigma_{nX}^{2}(\mathbf{c})}}\leq x\right)\left|\mathbf{C}=\mathbf{c}\right.\right]
\]
converges to $\Phi(x)$ for each $x$ and almost every $\mathbf{c}$,
where $\Phi$ is the cumulative distribution function of the standard
Normal random variable and $\mathbf{C}$ is a given sequence $(C_{i}:i=1,\ldots,n)$.
Let $P_{\mathbf{C}}$ denote the distribution of $\mathbf{C}$. Then
\begin{align*}
P(Z_{nX}\leq x) & \equiv P\left[\left(\sum_{i=1}^{N}\dfrac{f_{\mathbf{X},i}(\mathbf{X},\mathbf{C})}{\sqrt{\sigma_{nX}^{2}}}\leq x\right)\right]\\
 & =\int_{\mathbf{c}}P(Z_{nX}'\leq x|\mathbf{C}=\mathbf{c})dP_{\mathbf{C}}(\mathbf{c}).
\end{align*}

For a given $x$, the dominated convergence theorem is now applied
with $f_{n}(\mathbf{c})=P(Z_{nX}'\leq x|\mathbf{C}=\mathbf{c})$ and
the limit given by $f(\mathbf{c})=\Phi(x)=m$, where $m$ is some
constant that doesn't depend on $\mathbf{c}$. From the previous conditional
CLT result it follows that $f_{n}(\mathbf{c})$ converges to $f(\mathbf{c})$
pointwise for each $\mathbf{c}$. The next step is to find an integrable
function $g$, such that $f_{n}<g$ and $\int g(\mathbf{c})dP_{\mathbf{C}}(\mathbf{c})<\infty$.
The proof is then completed by choosing $g=1$. \end{proof}

We have now shown that $Z_{nY}$, $Z_{nX}$, and $Z_{nC}$ are asymptotically
normally distributed. We now show that the sum of the three processes
converges in distribution to a Normal random variable. Consider three
cases: (1) the three processes have the same rate of marginal convergence
in distribution, (2) one of the three processes converges faster than
the other two, and (3) two of the processes converge faster than the
third. In all three cases the rate of convergence for the sum will
be the slowest of the three marginal rates. In case (3), the limiting
distribution of the sum is determined entirely by the one process
that converges with a slower rate than the other two: the other two
processes will converge to constants (specifically to their expected
values of 0) when standardized by the slower rate; Slutsky's theorem
concludes the proof. We focus on case (1) below; case (2) follows
immediately by applying the proof below to the two processes that
converge at the same slower rate and applying Slutsky's to the third,
faster converging process.

For convenience, in order to show that the sum of the three dependent
processes also converges to Normal, define 
\[
C_{n}^{*}:=\sigma_{nY}^{2}+\sigma_{nX}^{2}+\sigma_{nC}^{2}.
\]
Note that $C_{n}^{*}$ is related to $C_{n}$ as follows: $C_{n}=O(n^{2}/C_{n}^{*})$.

\begin{lemma}[CLT for the sum of the three orthogonal processes]\label{lem:sumCLT}If
all three processes have the same marginal rate of convergence, then
\[
\dfrac{1}{\sqrt{C_{n}^{*}}}\left(f_{\mathbf{Y}}(\mathbf{Y},\mathbf{\mathbf{X}},\mathbf{C})+f_{\mathbf{\mathbf{X}}}(\mathbf{\mathbf{X},}\mathbf{C})+f_{\mathbf{C}}(\mathbf{C})\right)\rightarrow N(0,1).
\]
\end{lemma}

\begin{proof}[Proof of Lemma~\ref{lem:sumCLT}] Without the
loss of generality, we prove that $Z_{nX}+Z_{nC}\rightarrow N(0,2)$
and note that the general result for $(Z_{nY}+Z_{nX}+Z_{nC})$ follows
by applying a similar set of arguments.

Consider the following random vector $(Z_{nX},Z_{nC})$ taking values
in $\openr^{2}$. Let $F_{n}(x_{1},x_{2})\equiv P(Z_{nX}\leq x_{1},Z_{nC}\leq x_{2})$,
where $(x_{1},x_{2})\in\openr^{2}$. Let $\Phi^{2}(x_{1},x_{2})\equiv P(Z_{X}\leq x_{1})P(Z_{C}\leq x_{2})$,
for $Z_{X}\sim N(0,1)$ and $Z_{C}\sim N(0,1)$, that is, $\Phi^{2}(x_{1},x_{2})$
defines the CDF of the bivariate standard normal distribution, for
$(x_{1},x_{2})\in\openr^{2}$. The goal is to show that $F_{n}(x_{1},x_{2})\rightarrow\Phi^{2}(x_{1},x_{2})$,
for any $(x_{1},x_{2})\in\openr^{2}$. The convergence in distribution
for $Z_{nX}+Z_{nC}$ will follow by applying the Cramer and Wold Theorem
(1936).

Note that 
\begin{align*}
 & P(Z_{nX}\leq x_{1},Z_{nC}\leq x_{2})\\
= & P(Z_{nX}\leq x_{1}\left|Z_{nC}\leq x_{2}\right.)P(Z_{nC}\leq x_{2}).
\end{align*}
First, from the previous application of Stein's method, we have that
\[
P(Z_{nC}\leq x_{2})\rightarrow\Phi(x_{2}),
\]
where $\Phi(x_{2})\equiv P(Z_{C}\leq x_{2})$, $Z_{C}\sim N(0,1)$
and $x_{2}\in\openr^{2}$. Also note that 
\begin{eqnarray*}
 &  & P(Z_{nX}\leq x_{1}\left|Z_{nC}\leq x_{2}\right.)\\
 & = & \sum_{\mathbf{c}\in\mathcal{C}}P(Z_{nX}\leq x_{1}\left|\mathbf{C}=\mathbf{c}\right.)P(\mathbf{C}=\mathbf{c}\left|Z_{nC}\leq x_{2}\right.),
\end{eqnarray*}
where $\mathcal{C}$ denotes the support of $\mathbf{C}$, $Z_{nX}=\frac{1}{\sqrt{C_{n}^{*}}}f_{\mathbf{X}}(\mathbf{\mathbf{X},}\mathbf{C})$,
$Z_{nC}=\frac{1}{\sqrt{C_{n}^{*}}}f_{\mathbf{C}}(\mathbf{C})$ and
\[
P(\mathbf{C}=\mathbf{c}\left|Z_{nC}\leq x_{2}\right.)=\dfrac{P(\mathbf{C}=\mathbf{c})I(\left(1/\sqrt{C_{n}^{*}}\right)f_{\mathbf{C}}(\mathbf{c})\leq x_{2})}{P(\left(1/\sqrt{C_{n}^{*}}\right)f_{\mathbf{C}}(\mathbf{c})\leq x_{2})}.
\]
By another application of Stein's method, it was shown that 
\[
P(Z_{nX}\leq x_{1}\left|\mathbf{C}=\mathbf{c}\right.)\rightarrow\Phi(x_{2}),
\]
for any realization of $\mathbf{c}\in\mathcal{C}$. That is, we've
shown that the limiting distribution of $Z_{nX}$ conditional on $\mathbf{C=\mathbf{c}}$,
does not itself depend on the conditioning event $\mathbf{C}=\mathbf{c}$.
Applying Lemma 3, we finally conclude that $F_{n}(x_{1},x_{2})\rightarrow\Phi^{2}(x_{1},x_{2})$,
for any $(x_{1},x_{2})\in\openr^{2}$ and the result follows.\end{proof}

\section{Variance estimation}

In principle, the estimate of the variance of the $\hat{\psi}$ can always be obtained by the plug-in estimator of 
\[
\frac{1}{n^{2}}\sum_{i,j}E(f_{i}(\mathbf{O})f_{j}(\mathbf{O})).
\]
Alternatively, recalling the orthogonal decomposition used in the proof of Theorem 1, a variance estimate can be obtained
from the sum, scaled by $1/n^{2}$, of the three plug-in estimators of 
\begin{eqnarray*}
\sigma_{nY}^{2} & = & \sum_{i,j}E(f_{\mathbf{Y},i}(\mathbf{O})f_{\mathbf{Y},j}(\mathbf{O}))\\
\sigma_{nX}^{2} & = & \sum_{i,j}E(f_{\mathbf{X},i}(\mathbf{O})f_{\mathbf{X},j}(\mathbf{O}))\\
\sigma_{nC}^{2} & = & \sum_{i,j}E(f_{\mathbf{C},i}(\mathbf{O})f_{\mathbf{C},j}(\mathbf{O})).
\end{eqnarray*}

Note that contribution to these variances of any pair $i,j$ not in each others dependency neighborhoods will be $0$. Therefore, it is acceptable to sum only over pairs $i,j$ sharing a tie or a mutual contact in the underlying network. 

When dependence is due to direct transmission, that is under assumptions (A1), (A4), and (A5), $C$ is iid and $\sigma_{nC}^{2}$ can be estimated with the empirical distribution of $C$ and, under correct specification of both $m(\cdot)$ and $g(\cdot)$, the plug-in estimator is consistent:
\[
\frac{C_n}{n^{2}}\sum_{i,j}(\hat f_{i}(\mathbf{O}) \hat f_{j}(\mathbf{O})) \xrightarrow{p} \sigma^2. 
\]
This is similar to results for other doubly-robust estimators, where standard analytic variance estimators are generally consistent only under correct specification of both nuisance functionals. That the plug-in variance estimator $\frac{1}{n^{2}}\sum_{i,j}E(\hat f_{i}(\mathbf{O}) \hat f_{j}(\mathbf{O}))$  is unbiased for $\frac{1}{n^{2}}\sum_{i,j}E(f_{i}(\mathbf{O})f_{j}(\mathbf{O}))$ follows immediately from the assumption of correct specification. In the proof of Lemma 1 we demonstrated that  $var(\sum_{i,j}f_{\mathbf{V},i}(\mathbf{O})f_{\mathbf{V},j}(\mathbf{O}))=O(K_{max,n}/\sqrt{n})$ for $\mathbf{V}=\mathbf{C},\mathbf{X},$ or $\mathbf{Y}$; therefore each of these variances converges in probability to $0$ once scaled by $\frac{C_n}{n^{2}}$ and the proof of consistency is complete.

Note that we do not need
to know the true rate of convergence $\sqrt{C_{n}}$ to obtain a valid
estimate of the C.I. for $\psi$ in finite samples; this rate is captured by the number
of non-zero terms in the variance sums. That is, in finite samples we rely on the approximation of the distribution of $\hat \psi$ as $N(\psi,  \frac{1}{n^{2}}\sum_{i,j}E(\hat f_{i}(\mathbf{O}) \hat f_{j}(\mathbf{O})))$, which does not rely on explicit knowledge of $C_n$.

However, when latent variable dependence is present, that is under assumptions (A1) through (A3), no nonparametric estimator is available for $\sigma_{nC}^{2}$, and therefore the plug-in variance estimator would require a model for $p_C(\mathbf{c})$ in addition to $m(\cdot)$ and $g(\cdot)$. If such a model can be correctly specified then the results above hold and the plug-in estimator is consistent. Consistent variance esitmation is also available for the conditional estimand, for which $\frac{1}{n^{2}}\sum_{i,j}E(f_{i}(\mathbf{O})f_{j}(\mathbf{O}))$ depends only on $m(\cdot)$ and $g(\cdot)$. 

\part{Simulations}

All simulation and estimation was carried out in $\mathbf{R}$ language
\citep{r} with packages \emph{$\mathtt{simcausal}$} \citep{R-simcausal}
and \emph{$\mathtt{tmlenet}$} \citep{R-tmlenet}. The full $\mathbf{R}$
code for this simulation study is available in a separate github repository (\verb!github.com/osofr/Ogburn_etal_simulations!). 
{\citet{sofryginTechreport, sofrygin2017conducting, sofrygin2018single} provide additional details on implementation, computation, and simulations for asymptotic regimes with a bounded number of ties per node and with no latent variable dependence.

The simulations were repeated for community sizes of $n=500$, $n=1,000$
and $n=10,000$. The estimation was repeated by sampling $1,000$
such datasets, conditional on the same network (sampled only once
for each sample size). For the simulations with dependence due to
direct transmission,\textbf{ }the baseline covariates were independently
and identically distributed. The probability of success for each $Y_{i}$
was a logit-linear function of $i$'s exposure $X_{i}$ (indicator
of receiving the economic incentive), the baseline covariates $C_{i}$
and the three summary measures of $i$'s friends exposures and baseline
covariates. In particular, we also assumed that the probability of
maintaining gym membership increased on a logit-linear scale as a
function of the following network summaries: the total number of $i$'s
friends who were exposed $(\sum_{j:A_{ij}=1}X_{j})$, the total number
of $i$'s friends who were physically active at baseline $(\sum_{j:A_{ij}=1}PA{}_{j})$
and the product of the two summaries $(\sum_{j:A_{ij}=1}X{}_{j}\times\sum_{j:A_{ij}=1}PA{}_{j})$. The summary measures and the outcome regression model were correctly specified, but we do not know (and therefore did not a priori correctly specify a model for) the true density of $h$.
The economic incentive to attend local gym had a small direct effect
on each individual who was not physically active at baseline and no
direct effect on those who were already physically active. However,
physically active individuals were more likely to maintain gym membership
over the follow-up period if they had at least one physically active
friend at baseline. We repeated these simulations with the addition
of latent variable dependence, which we introduced by generating unobserved
latent variables for each node which affected the node's own outcome
as well as the outcomes of its friends.

In addition to the preferential attachment network model with both latent variable dependence and dependence due to direct transmission (results in main text), we also simulated under dependence due to
direct transmission only. We estimated the marginal parameter $E\left[\bar{Y}^{*}_{n}\right]$
and compared three different estimators of the asymptotic variance
and the coverage of the corresponding confidence intervals. First,
we looked at the naive plug-in i.i.d. estimator (``\emph{IID Var}'')
for the variance of the influence curve which treated observations
as if they were i.i.d. Second, we used the plug-in variance estimator
based on the efficient influence curve which adjusted for the correlated
observations (``\emph{dependent IC Var}'') \citep{sofryginTechreport}. Finally, we used the parametric bootstrap
variance estimator (``\emph{bootstrap Var}'') described in Section 3.5. The simulation results showing the mean length and coverage of these three CI types
are shown in Figure \ref{fig:CIres.EY.prefattach}.

\begin{figure} [h]
\begin{centering}
\includegraphics[scale=0.7]{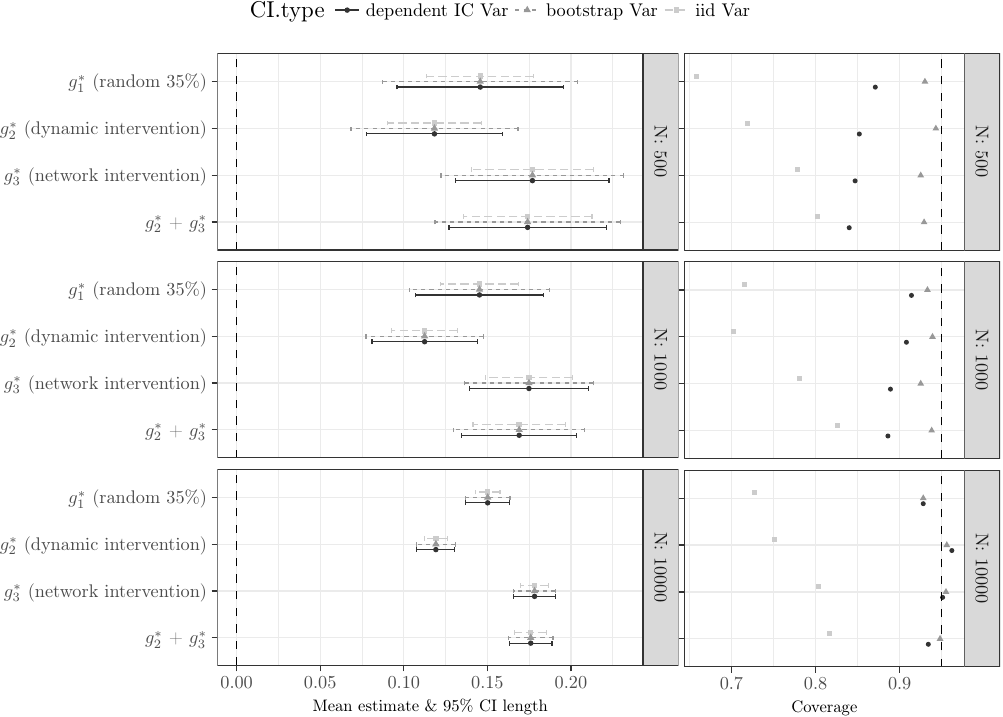} 
\par\end{centering}

\protect\protect\caption{Mean 95\% CI length (left panel) and coverage (right panel) for the
TMLE in preferential attachment network with dependence due to direct
transmission, by sample size, intervention and CI type. \label{fig:CIres.EY.prefattach}}
\end{figure}

Results from simulations with dependence due to direct transmission
show that conducting inference while ignoring the nature of the dependence
in such datasets generally results in anticonservative variance estimates
and under-coverage of CIs, which can be as low as 50\% even for very
large sample sizes (``\emph{IID Var}'' in the right panel of Figure
\ref{fig:CIres.EY.prefattach}). The
CIs based on the dependent variance estimates (``\emph{dependent
IC Var}'') obtain nearly nominal
coverage of 95\% for large enough sample sizes, but can suffer in
smaller sample sizes due to lack of asymptotic normality and near-positivity
violations. Notably, the CIs based on the parametric bootstrap variance
estimates provide the most robust coverage for smaller sample sizes,
while attaining the nominal 95\% coverage in large sample sizes for
nearly all of the simulation scenarios (``\emph{bootstrap Var}''). The apparent robustness of the parametric bootstrap method for inference in small sample sizes,
even as low as $n=500$, was one of the surprising finding of this
simulation study. Future work will explore the assumptions under which
this parametric bootstrap works and its sensitivity towards violations
of those assumptions.

We also simulated social networks from the small world network model \citep{watts1998collective} with a rewiring probability of $0.1$. The results of these simulations are in Figures \ref{fig:CIres.EY.smwld} and \ref{fig:CIres.EY.smwld.COND.DEP.Y}.

\begin{figure}
\begin{centering}
\includegraphics[scale=0.8]{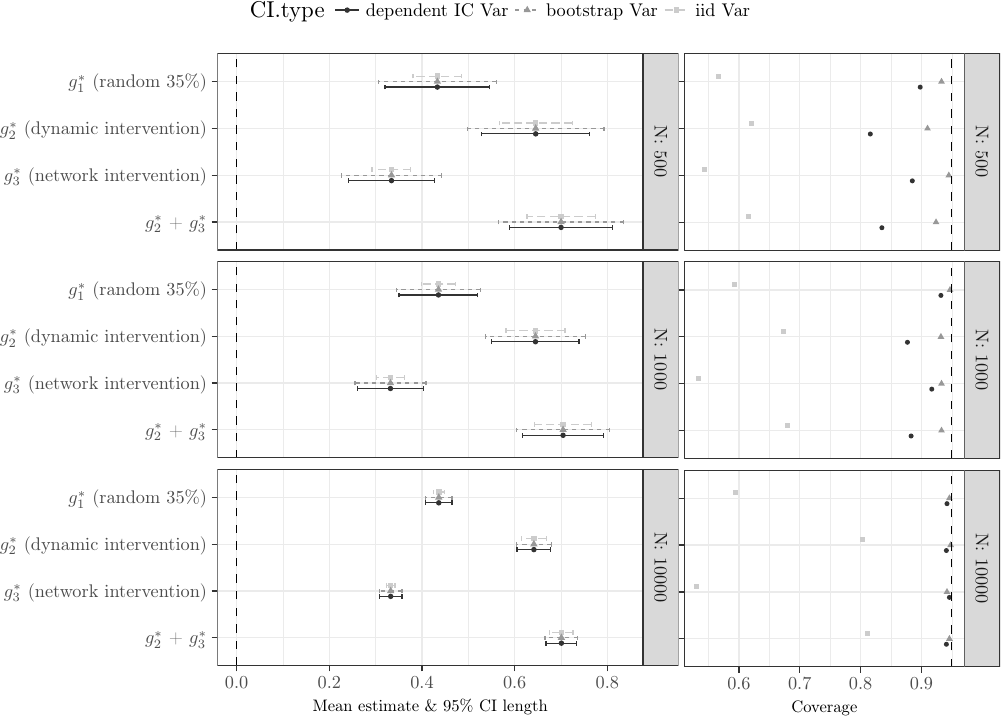} 
\par\end{centering}

\protect\protect\caption{Mean 95\% CI length (left panel) and coverage (right panel) for the
TMLE in small world network\textbf{ }with dependence due to direct
transmission, by sample size, intervention and CI type. Results are
shown for the estimates of the average expected outcome under four
hypothetical interventions ($g_{1}^{*}$, $g_{2}^{*}$, $g_{3}^{*}$
and $g_{2}^{*}+g_{3}^{*}$). \label{fig:CIres.EY.smwld}}
\end{figure}

\begin{figure}
\begin{centering}
\includegraphics[scale=0.8]{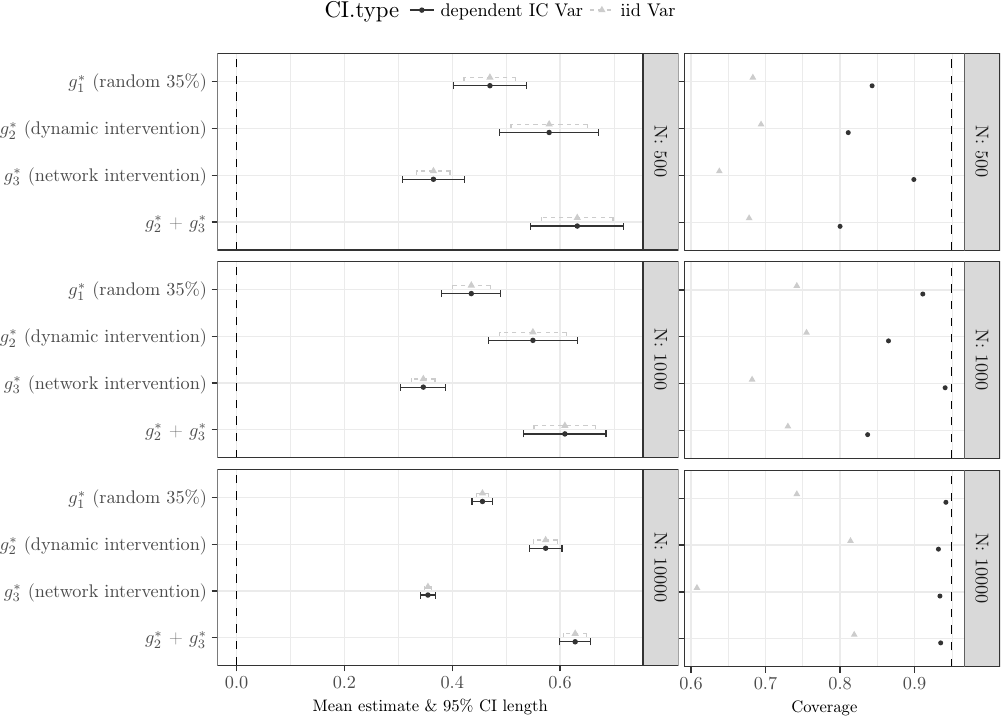} 
\par\end{centering}

\protect\protect\caption{Mean 95\% CI length (left panel) and coverage (right panel) for the
TMLE in small world network with latent variable dependence, by sample
size, intervention and CI type. Results are shown for the estimates
of the average expected outcome under four hypothetical interventions
($g_{1}^{*}$, $g_{2}^{*}$, $g_{3}^{*}$ and $g_{2}^{*}+g_{3}^{*}$).
\label{fig:CIres.EY.smwld.COND.DEP.Y}}
\end{figure}

We examined the empirical distribution of the transformed TMLEs, comparing
their histogram estimates to the predicted normal limiting distribution,
with the results shown in Figure \ref{fig:hist.TMLE.EY}, where the
histogram plots are displayed by sample size (horizontal axis) and
the intervention type (vertical axis). The estimates were first centered
at the corresponding true parameter values and then re-scaled by their
corresponding true standard deviation (SD). We note that our results
indicate that the estimators converge to their normal theoretical
limiting distribution, even in networks with power law node degree
distribution, such as the preferential attachment network model, as
well as in the densely connected networks obtained under the small
world network model. The results shown in Figure \ref{fig:hist.TMLE.EY}
were generated from simulations with dependence due to direct transmission;
simulations with latent variable dependence (not shown) evinced similar
approximate normality.

\begin{figure}
\centering{}\includegraphics[scale=0.25]{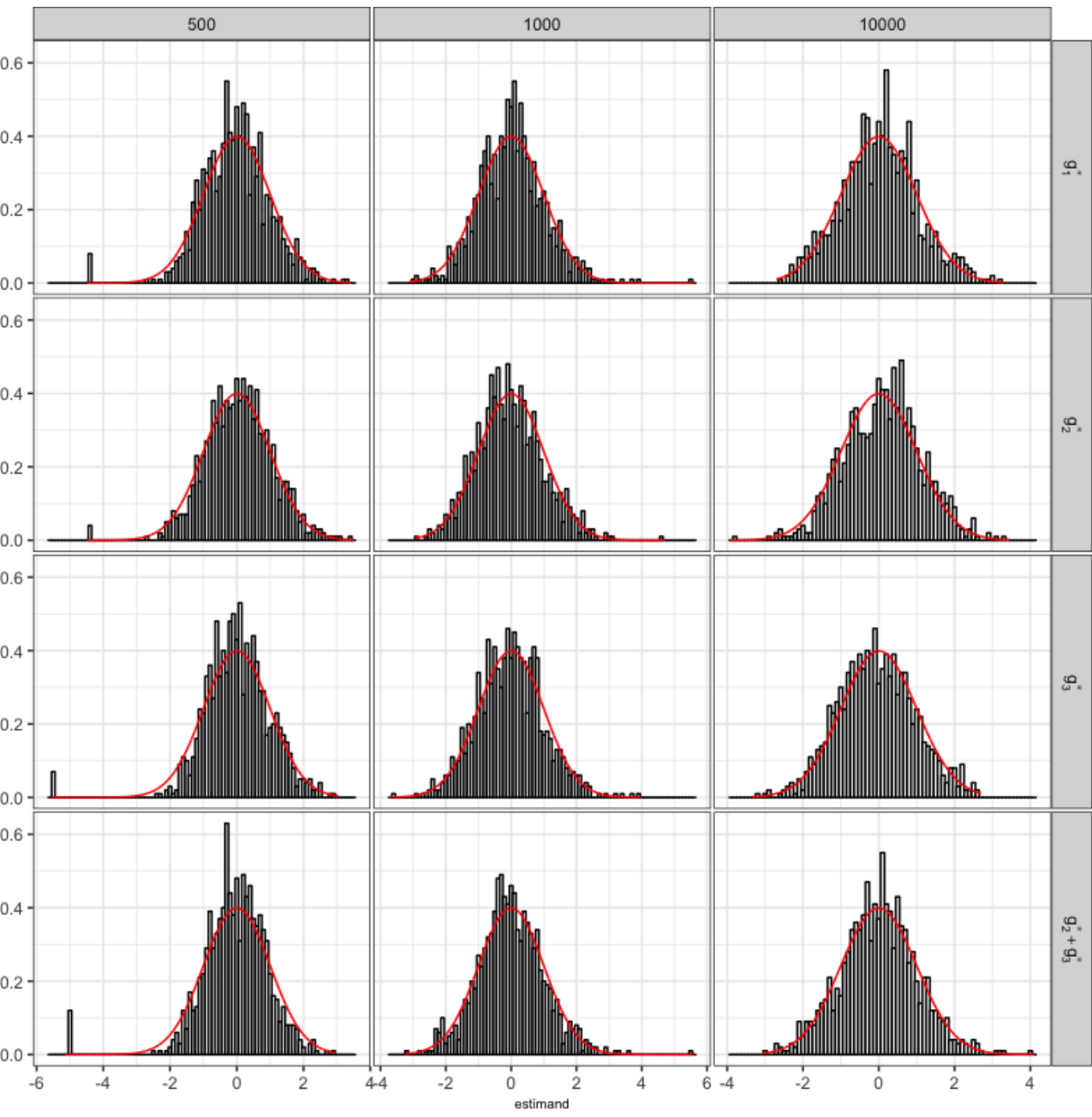}\includegraphics[scale=0.25]{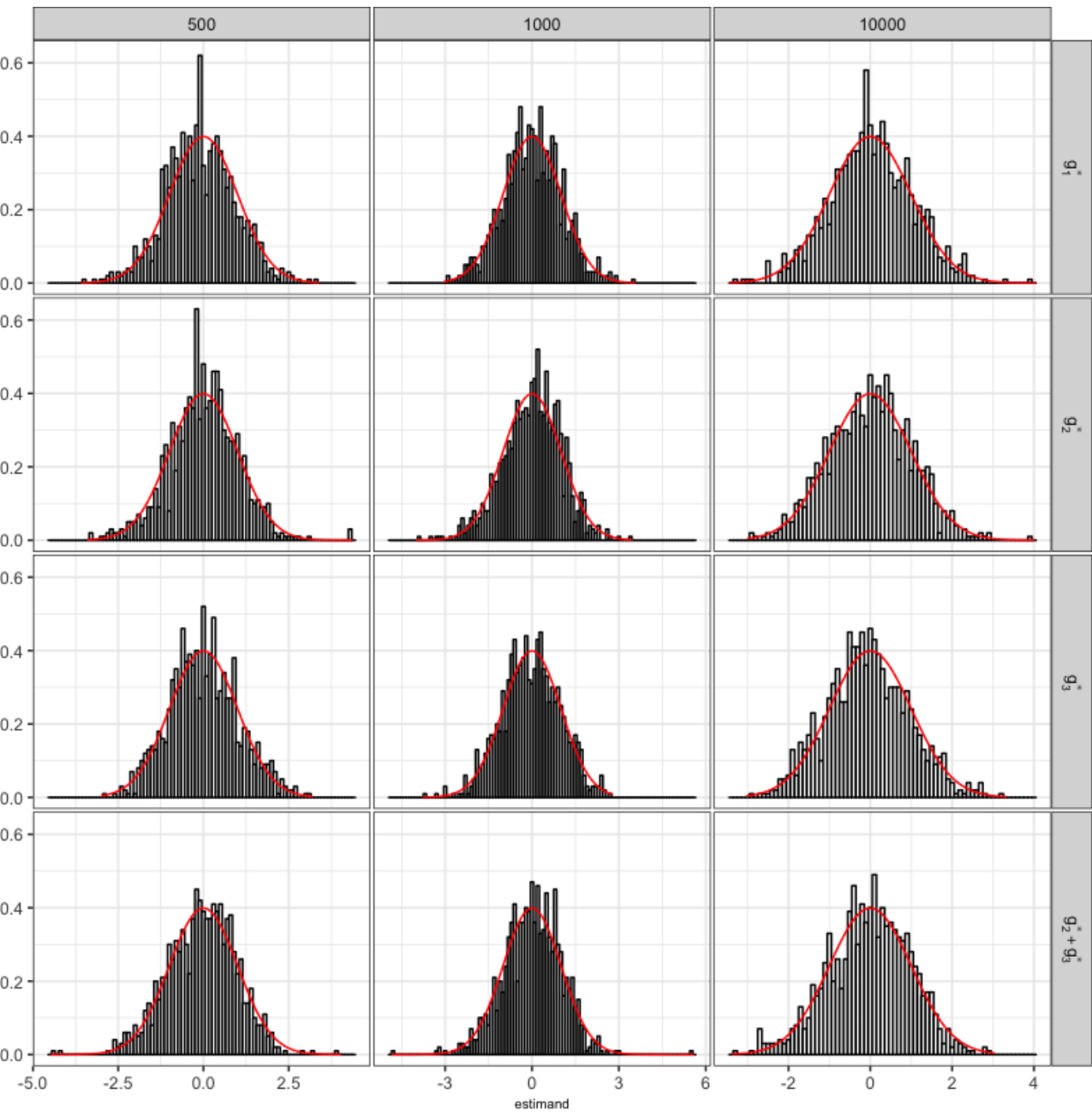}\protect\protect\caption{Comparing re-scaled empirical TMLE distributions (black) to their
theoretical normal limit (red) with varying sample size (x-axis) and
intervention type (y-axis). TMLEs were centered at the truth and then
re-scaled by true SD. Results shown for the preferential attachment
network (left) and the small world network (right).\label{fig:hist.TMLE.EY}}
\end{figure}

\part{ Glossary of Notation}

$\mathbf{A}$ with entries $A_{ij}\equiv I\left\{ \mbox{subjects }i\mbox{ and }j\mbox{ share a tie}\right\} $ is the adjacency matrix for the network.

\vspace{3mm}

\noindent $K_{i}=\sum_{j=1}^{n}A_{ij}$, that is, $K_{i}$
is the degree of node $i$, or the number of individuals sharing a tie
with individual $i$. 

\vspace{3mm}

\noindent 
$F_{i}={j:A_{ij}=1}$ is the set of nodes with with node $i$ shares a tie (node $i$'s "friends").

\vspace{3mm}

\noindent 
$C_{i}$ is covariates

\vspace{3mm}

\noindent 
$X_{i}$ is exposure

\vspace{3mm}

\noindent 
$Y_{i}$ is outcome

\vspace{3mm}

\noindent 
$s_{C}$ is a summary function of $\mathbf{C}$ upon which $X$ and $Y$ depend.

\vspace{3mm}

\noindent 
$s_{X}$ is a summary function of $\mathbf{X}$ upon which $Y$ depends.

\vspace{3mm}

\noindent 
$W_{i}=s_{C,i}\left(\left\{ C_{j}:A_{ij}=1\right\} \right)$  

\vspace{3mm}

\noindent 
$V_{i}=s_{X,i}\left(\left\{ X_{j}:A_{ij}=1\right\} \right)$

\vspace{3mm}

\noindent 
$O_{i}=(C_{i},W_{i},X_{i},V_{i},Y_{i})$

\vspace{3mm}

\noindent 
$x_{i}^{*}$ represents a user-specified intervention value of $X_{i}$.

\vspace{3mm}

\noindent 
$Y_{i}(\mathbf{x}^{*})$, shorthand $Y_{i}^{*}$, denotes the potential or counterfactual outcome of individual $i$ in a hypothetical
world in which $P(\mathbf{X}=\mathbf{x}^{*})=1$. 

\vspace{3mm}

\noindent 
$V_{i}(\mathbf{x}^{*})$, shorthand $V_{i}^{*}$, is equal to $s_{X,i}(\mathbf{x}^{*})$ and
is a counterfactual random variable in a hypothetical world in which
$P(\mathbf{X}=\mathbf{x}^{*})=1$\textbf{. } 

\vspace{3mm}

\noindent 
$\bar{Y}^{*}_{n}=\frac{1}{n}\sum_{i=1}^{n}Y_{i}^{*}$.

\vspace{3mm}

\noindent 
 $p_{C}(\mathbf{c})=P\left(\mathbf{C}=\mathbf{c}\right)$

\vspace{3mm}

\noindent 
 $g(\mathbf{x}|\mathbf{w})=P\left(\mathbf{X}=\mathbf{x}\mid\mathbf{W}=\mathbf{w}\right)$

\vspace{3mm}

\noindent 
$g_{i}(x|w)=P\left(X_{i}=x|W_{i}=w\right)$

\vspace{3mm}

\noindent 
$p_{Y}(\mathbf{y}|\mathbf{v,w})=P\left(\mathbf{Y=y}|\mathbf{V=v,W=w}\right)$

\vspace{3mm}

\noindent 
$p_{Y,i}(y|v)=P\left(Y_{i}=y|V_{i}=v,W_i=w\right)$

\vspace{3mm}

\noindent 
$h_{i}(v,w)=P\left(V_{i}=v,W_i=w\right)$ 

\vspace{3mm}

\noindent 
$h_{i,x^{*}}(v,w)=P\left(V_{i}^{*}=v,W_i=w\right)$

\vspace{3mm}

\noindent 
$m(v,w)=\sum_{y}y\,p_{Y}(y|v,w)$ is the conditional expectation of $Y$
given $V=v,W=w$.

\vspace{3mm}

\noindent 
 $\bar{h}(v_{i},w_i)=\frac{1}{n}\sum_{j=1}^{n}h_{j}(v_{i},w_i)$

\vspace{3mm}

\noindent 
 $\bar{h}_{x^{*}}(v_{i},w_i)=\frac{1}{n}\sum_{j=1}^{n}h_{j,x^{*}}(v_{i},w_i)$

\vspace{3mm}



\noindent 
$D_n(\mathbf{o})$ is the nonparametric influence function under assumptions (A1) and (A4).

\vspace{3mm}

\noindent 
$D_n^C(\mathbf{o})$ is an influence function conditional on $\mathbf{C}=\mathbf{c}$.

\vspace{3mm}

\noindent 
$K_{max,n}=max_{i}\{K_{i}\}$

\vspace{3mm}

\noindent 
$\sqrt{C_{n}}$ is the rate of convergence in Theorem 1.

\begingroup
\setstretch{1.0}

\bibliographystyle{jasa}
\bibliography{references}
\endgroup

\end{document}